\documentclass[12pt]{article}
\usepackage[utf8]{inputenc}
\usepackage[english]{babel}
 \usepackage{amsmath}
\usepackage{amsthm}
\usepackage{amssymb}
\usepackage{hyperref}
\usepackage{braket}
\usepackage{color}
\usepackage{comment}
\usepackage{graphicx}
\usepackage[margin=2.5cm]{geometry}
 \usepackage[noblocks]{authblk}
 \usepackage[sort&compress,numbers]{natbib}
\usepackage[noabbrev]{cleveref}

\newcommand{\epsi}{\varepsilon}
\newcommand{\E}{{\mathrm{e}}}
\newcommand{\I}{\mathrm{i}}
 \newcommand{\R}{ \mathbb{R} }
  \newcommand{\Sph}{ \mathbb{S} }

\newcommand{\N}{ \mathbb{N} }

\newcommand{\D}{\mathrm{d}}

\begin{document}

\allowdisplaybreaks

\theoremstyle{plain}
\newtheorem{thm}{Theorem}
\newtheorem{lem}[thm]{Lemma}
\newtheorem{cor}[thm]{Corollary}
\newtheorem{prop}[thm]{Proposition}

\theoremstyle{definition}
\newtheorem{defi}[thm]{Definition}
\newtheorem{rmk}[thm]{Remark}
\newtheorem{assumption}[thm]{Assumption}

\title{The BCS Critical Temperature at High Density }
\author{Joscha Henheik\footnote{\href{mailto:joscha.henheik@ist.ac.at}{joscha.henheik@ist.ac.at} } \\ IST Austria, Am Campus 1, 3400 Klosterneuburg, Austria}
\maketitle
\begin{abstract}
We investigate the BCS critical temperature $T_c$ in the high--density limit and derive an asymptotic formula, which strongly depends on the behavior of the interaction potential $V$ on the Fermi--surface.  Our results include a rigorous confirmation for the behavior of $T_c$ at high densities proposed by Langmann, Triola, and Balatsky (Phys.~Rev.~Lett.~122, 2019) and identify precise conditions under which superconducting domes arise in BCS theory.
\\~ \\
{
	\bfseries
	Keywords:
}
BCS theory, Critical Temperature, Superconducting Domes
\\ 
{
	\bfseries
	Mathematics subject classification: 
}
81Q10, 46N50, 82D55
\end{abstract}

\section{Introduction}	
The Bardeen--Cooper--Schrieffer (BCS) gap equation \cite{bcs}
\begin{equation}
\label{gapeq} \Delta(p) = -\frac{1}{(2\pi)^{3/2}} \int_{\R^3} \hat{V}(p-q) \frac{\Delta(q)}{E_{\Delta, \mu}(q)} \tanh\left(\frac{E_{\Delta, \mu}(q)}{2T}\right) \D q\,,
\end{equation} 
with dispersion relation $E_{\Delta, \mu}(p) = \sqrt{(p^2-\mu)^2+ \vert \Delta(p) \vert^2}$, has played an important role in physics since its introduction. The function $\Delta$ is interpreted as the order parameter describing paired fermions (Cooper pairs) interacting via the local pair potential $2V$, which we assume to be integrable, i.e.~$V \in L^1(\R^3)$. In this case, $\hat{V}(p) = (2\pi)^{-3/2} \int_{\R^3} V(x) \E^{- \I p\cdot x} \D x$ denotes its Fourier transform. The positive parameters $T$ and $\mu$ are the temperature and the chemical potential, respectively, where the latter controls the density of fermions. Whenever the temperature $T$ is below a certain critical temperature $T_c$ (see Definition~\ref{def:1}), the gap equation~\eqref{gapeq} admits non-trivial solutions, above it does not. Physically, this corresponds to the the system being in a superconducting state ($T<T_c$) or a normal state ($T\ge T_c$). 

BCS theory has previously been studied in the weak--coupling limit \cite{hs081,fhns} and low--density limit \cite{hs08,lauritsen}. 
In the weak--coupling limit one considers a potential $\lambda V$ for a fixed potential $V$ for small coupling constant $\lambda \to 0$.  
In this limit, it was shown by Hainzl and Seiringer \cite{hs081} that the critical temperature satisfies  $T_c \sim A \exp(-B/\lambda)$ for explicit constants $A,B >0$. 
In the low--density limit, $\mu \to 0$,  it is shown, again by Hainzl and Seiringer \cite{hs08}, that $T_c \sim \mu A \exp(-B/\sqrt{\mu})$ for some (different) explicit constants $A,B >0$ (see \Cref{eq:lowdensityTc}).

In this paper we are interested in the critical temperature for the existence of non-trivial solutions of the BCS gap equation \eqref{gapeq} in the high--density limit, i.e. $\mu \to \infty$. 
Studying the high--density limit of the critical temperature is especially relevant for explaining superconducting domes \cite{srtio, cuprates, pnictides, heavyfermion,bandins, graphene}, i.e.~a non-monotonic $T_c(\mu)$ exhibiting a maximum value at finite $\mu$ and going to zero for large $\mu$. In a recent paper \cite{langmann}, the authors claim the \textit{ubiquity of superconducting domes in BCS theory}, but only for pure s--wave superconductivity (i.e.~angular momentum $\ell = 0$, see Remark~\ref{rmk:langmann}). Their result disproves the conventional wisdom, that the presence of a superconducting dome necessarily indicates some kind of exotic superconductivity, e.g.~resulting from competing orders. BCS theory containing a non--monotonic behavior of $T_c(\mu)$ is in particular relevant for understanding superconducting domes in doped band insulators \cite{bandins} and magic--angle graphene \cite{graphene}, where no competing orders occur, and thus a more conventional explanation is desirable. 

There is a simple physical picture arising from an interplay of length scales, that explains the ubiquitous appearance of superconducting domes (see \cite{langmann}). If the effective range $\xi$ of the interaction is much smaller than the mean interparticle distance $\mu^{-1/2}$, i.e.~$\xi \ll \mu^{-1/2}$, the critical temperature $T_c$ increases by increasing $\mu$ as predicted by standard BCS theory \cite{bcs} and rigorously justified in \cite{hs08}. At high densities, i.e.~if $\xi \gg \mu^{-1/2}$, the pairing of electrons near the Fermi surface (with approximately opposite momenta), which is responsible for the superconducting behavior, becomes weaker with increasing $\mu$ due to the decay of the interaction in Fourier space, suppressing $T_c$ towards zero. Therefore, at some intermediate density, where $\xi \sim \mu^{-1/2}$, a superconducting dome arises.  This simple argument is reflected in our results by the presence of the operator $\mathcal{V}_{\mu}$, defined in \Cref{eq:vmu}, acting on functions on the (rescaled) Fermi surface.

Our results in Section \ref{sec:results} are threefold: first, we confirm a proposed asymptotic formula from \cite{langmann} for the critical temperature at high densities for s--wave superconductivity (to leading order) by proving a more general result for radially symmetric interaction potentials $V$ (Theorem~\ref{thm:1}); second, we provide a counterexample, showing that the assumptions on $V$ from \cite{langmann} are not quite sufficient to conclude a non--monotonic behavior of $T_c$ and need to be slightly strengthened (Proposition~\ref{prop:1}); third, we use these strengthened assumptions to improve the asymptotics obtained in Theorem~\ref{thm:1} to second order with the aid of perturbation theory, and obtain an analogous formula to the ones proven in the weak--coupling and low--density limit (Theorem~\ref{thm:2}). All proofs are given in Section~\ref{sec:proofs}.

\section{Main Results} \label{sec:results}

\subsection{Preliminaries}
 It was proven in \cite{hhss} (see also \cite{hs15} for a more recent review) that the critical temperature for the existence of non--trivial solutions of the BCS gap equation \eqref{gapeq} can be characterized as follows. 
\begin{defi}{\rm (Critical Temperature)} \label{def:1} \\
Let $\mu > 0$, $V \in L^{3/2}(\R^3)$ be real--valued and $K_{T,\mu}$ denote the multiplication operator in momentum space
\begin{equation*}
K_{T,\mu}(p) = \frac{\vert p^2 - \mu\vert}{\tanh\left(\frac{\vert p^2 - \mu \vert}{2T} \right)}\,.
\end{equation*}
The critical temperature for the BCS gap equation \eqref{gapeq} is given by
\begin{equation*}
T_c = \inf \left\{T > 0 \mid K_{T,\mu}(p) + V(x) \ge 0 \right\} \,.
\end{equation*}
\end{defi}
\noindent One might think of the operator $K_{T,\mu}(p)+ V(x)$ as the Hessian in the BCS functional of superconductivity at a normal state (see \cite{hs15}), where the positivity corresponds to the ``stability" of this normal state, which is directly related to the existence of a non--trivial solution of the BCS gap equation \eqref{gapeq}. 
Note that the continuous spectrum of $K_{T,\mu}$ starts at $2T$ and thus $T_c$ is well defined by Sobolev's inequality \cite[Thm.~8.3]{liebloss} since $K_{T,\mu} \sim p^2$ for large~$\vert p\vert$.

Moreover, note that $K_{T,\mu}$ takes its minimum value $2T$ on the $\mathrm{codim-}1$ submanifold $\set{p^2 = \mu}$. Thus, similarly to the weak coupling situation \cite{fhns} and as pointed out by Laptev, Safronov, and Weidl \cite{laptev} (see also \cite{hs10}), the spectrum of $K_{T,\mu}+ V$ is mainly determined by the behavior of $V$ near $\set{p^2 = \mu}$, i.e.~the Fermi sphere. More precisely, as emphasized in the introduction, a crucial role for the investigation of $T_c$ in the high--density limit is played by the (rescaled) operator  $\mathcal{V}_{\mu} : L^2(\mathbb{S}^{2}) \to L^2(\mathbb{S}^{2})$ where
\begin{equation} \label{eq:vmu}
\left(\mathcal{V}_{\mu} u\right)(p) = \frac{1}{(2\pi)^{3/2}} \int_{\Sph^{2}} \hat{V}(\sqrt{\mu} (p-q)) u(q)\, \D\omega(q) \,.
\end{equation}
Here $\D \omega$ denotes the uniform (Lebesgue) measure on the unit sphere $\Sph^2$. The pointwise evaluation of $\hat{V}$ (and thus also on a $\mathrm{codim-}1$ submanifold) is well defined since $\hat{V}$ is continuous for $V \in L^1(\R^3)$. See Remark \ref{rmk:3} for a discussion of the assumption $V \in L^1(\R^3)$ (cf.~also \cite{cueninmerz}). The lowest eigenvalue of $\mathcal{V}_{\mu}$, which we denote by
\begin{equation*}
e_{\mu} = \mathrm{inf\,spec}\, \mathcal{V}_{\mu}
\end{equation*}
will be of particular importance. Note, that $\mathcal{V}_{\mu}$ is a compact operator (so $e_{\mu} \le 0$), which is in fact trace class (see the argument above Equation (3.2) in \cite{fhns}) with
\begin{equation*}
\mathrm{tr}(\mathcal{V}_{\mu}) = \frac{1}{2\pi^2} \int_{\R^3} V(x) \D x\,.
\end{equation*} 
The case $e_{\mu} < 0$ will be important for our main results as it corresponds to an attractive interaction between (some) electrons on the Fermi sphere. Since $\mathcal{V}_{\mu}$ is trace class, a sufficient condition for $e_{\mu}< 0$ is that the trace of $\mathcal{V}_\mu$ is negative, i.e.~$\int V < 0$. Moreover, by considering a trial function that is concentrated on two small sets on the rescaled Fermi sphere $\Sph^{2}$ separated by a distance $\vert p \vert < 2$, one can easily see that $e_{\mu}<0$ if $\vert\hat{V}(p)\vert>\hat{V}(0)$ for some $\vert p\vert < 2 \sqrt{\mu}$. 

In this work, we restrict ourselves to the special case of radial potentials $V$ depending only on $\vert x \vert$, where the spectrum of $\mathcal{V}_{\mu}$ can be determined more explicitly (see, e.g., Section~2.1 in \cite{fhns}). Indeed, if $V$ is radially symmetric, the eigenfunctions of $\mathcal{V}_{\mu}$ are spherical harmonics and the corresponding eigenvalues are 
\begin{equation} \label{eq:eigenvalues}
\frac{1}{2\pi^2} \int_{\R^3} V(x) \left(j_\ell(\sqrt{\mu}\vert x \vert)\right)^2 \D x\,,
\end{equation}
with $\ell \in \mathbb{N}_0$ and where $j_\ell$ denotes the $\ell^{\mathrm{th}}$--order spherical Bessel function. A few important properties of the spherical Bessel functions used in our proofs are collected in Proposition~\ref{prop:2}. By \Cref{eq:eigenvalues}, the lowest eigenvalue $e_\mu$ is thus given by
\begin{equation*}
e_\mu = \frac{1}{2\pi^2} \, \inf_{\ell \in \mathbb{N}_0} \, \int_{\R^3} V(x) \left(j_\ell(\sqrt{\mu}\vert x \vert)\right)^2 \D x \,.
\end{equation*} 
If additionally $\hat{V} \le 0$, the minimal eigenvalue is attained for the constant eigenfunction (i.e.~the spherical harmonic with $\ell = 0$) by the Perron--Frobenius Theorem~and we thus have the more concrete expression 
\begin{equation} \label{eq:emuangmom0}
e_{\mu} = \frac{1}{2\pi^2} \int_{\R^3} V(x) \left(\frac{\sin(\sqrt{\mu}\vert x \vert)}{\sqrt{\mu}\vert x \vert}\right)^2 \D x\,.
\end{equation} 
We refer to Remark \ref{rmk:2} for a discussion of the radiality assumption on $V$.

\subsection{Results}
As desribed in the introduction, our results are threefold: First, we show an asymptotic formula for radial potentials $V$ (Theorem~\ref{thm:1}), including the rigorous confirmation of the result from \cite{langmann} to leading order. Afterwards, we provide a counterexample showing that the assumptions made in \cite{langmann} are not quite sufficient to conclude a non--monotonic behavior of $T_c$, i.e.~a superconducting dome (Proposition~\ref{prop:1}). Finally, by slightly strengthening the assumptions on $V$, we provide an asymptotic formula for the critical temperature valid to second order (Theorem~\ref{thm:2}). All proofs are given in Section~\ref{sec:proofs}.
\begin{thm} \label{thm:1}
Let $V \in L^1(\R^3) \cap L^{3/2}(\R^3)$ be real--valued and radially symmetric. Assume that there exists $\mu_0 > 0$ such that for all $\mu \ge \mu_0$ we have $e_{\mu}< 0$. Then $T_c(\mu) > 0$ for all sufficiently large $\mu$ and 
\begin{equation}
\lim\limits_{\mu \to \infty} \sqrt{\mu} \, e_{\mu} \, \log\frac{\mu}{T_c} = -1\,. \label{eq:thm1}
\end{equation}
\end{thm}
\noindent  Or in other words, we have the asymptotic behavior
\begin{equation}
T_c = \mu \,  \E^{(1+o(1))/(\sqrt{\mu}e_{\mu})}  \label{eq:leadingorder}
\end{equation}
in the limit of large $\mu$. Note, that the right hand side is the same formula as in the weak--coupling case \cite{fhns, hs081} but we have coupling parameter $\lambda = 1$.
\begin{rmk}{\rm (Connection to the result from \cite{langmann})} \label{rmk:langmann}\\
Assume that $V \in L^1(\R^3) \cap L^{3/2}(\R^3)$ is real--valued, radially symmetric and additionally satisfies $\hat{V} \le 0$ and $\hat{V}(0)<0$ (the latter implies that $e_\mu < 0$ for all $\mu >0$). Note that these conditions, which are identical to the ones required in \cite{langmann}, are included in the more general conditions of Theorem~\ref{thm:1}. Then we have, using the notation from \cite{langmann}, that
\begin{equation*} 
\sqrt{\mu} e_{\mu} = \frac{\sqrt{\mu}}{2 \pi^2} \int_{\R^3} V(x)  \frac{\sin^2(\sqrt{\mu} \vert x \vert)}{\mu \vert x \vert^2} \D x = \frac{1}{4\pi^2} \frac{f_{-2V}(4 \mu)}{4\sqrt{\mu}} =: - \lambda\,,
\end{equation*}
where the first equality follows by \Cref{eq:emuangmom0} and after inserting the definition of the function $f_{-2V}$ from \cite{langmann}, the second equality is a simple computation using Fubini. 
By means of Theorem~\ref{thm:1}, we thus confirmed the validity of Equation (6) in \cite{langmann} in the high--density limit to leading order, i.e.
\begin{equation*}
T_c = \mu \, \E^{-(1+o(1))/\lambda}\,.
\end{equation*}
In full generality, the asymptotic formula proposed in Equation (6) in \cite{langmann} reads
\begin{equation*}
T_c = \frac{2 \E^\gamma}{\pi} \mu \exp\left(- \frac{1}{\lambda} + \sum_{n=0}^{\infty} a_n \lambda^n \right)\,,
\end{equation*}
where $\gamma \approx 0.577$ denotes the Euler--Mascheroni constant and $(a_n)_{n \ge 0}$ is a sequence of explicit constants determined by an iterative procedure. The quantity $\lambda$ is understood as an intrinsic small parameter which encodes either a weak--coupling, low--density, or high--density limit, or an appropriate combination.
\end{rmk}

\noindent  In order to obtain a meaningful asymptotic formula of the critical temperature at high densities in a rigorous way, the question to be addressed now is the behavior of $\sqrt{\mu} e_\mu$ in the limit $\mu \to \infty$.  In the following proposition we present a special family of interaction potentials $(V_\alpha)$ showing that the conditions of Theorem~\ref{thm:1} (which include the more restricted conditions from \cite{langmann}) not necessarily lead to a non--monotonic behavior of $T_c$ as claimed in \cite{langmann}, since $\vert \sqrt{\mu} e_\mu \vert \gg \log(\mu)^{-1}$ in the limit of large $\mu$ for this family of potentials. More precisely, the $L^{3/2}(\R^3)$--condition, which essentially concerns the behavior of the interaction potential near the origin, is not quite sufficient to obtain a dome--shaped behavior of $T_c(\mu)$. Since the potentials $(V_\alpha)$ are perfectly well behaved away from the origin and decay rapidly at infinity, they illustrate the significance of the behavior of interaction potentials near the origin for the asymptotics of the critical temperature. It is natural that the critical temperature is sensitive to the short range behavior of the interaction potential, since the interparticle distance as the physically relevant length scale that depends on the particle density tends to zero in the high--density limit. 
\begin{prop}~\label{prop:1}
Let $\alpha \in (1/3,1/2)$ and set
\begin{equation*}
V_\alpha (x) = - \frac{\exp(-\vert x \vert)}{\vert x \vert^2 \left(  \log^2(\vert x \vert)+ 1\right)^{\alpha}}\,.
\end{equation*}
Then the critical temperature $T_c$ associated with $K_{T,\mu} + V_\alpha$ approaches infinity as $\mu \to \infty$. 
\end{prop}
\noindent Our observations from Proposition~\ref{prop:1} lead to the following definition of ``admissible potentials", that are slightly better behaved at the origin and, in particular, allow for an analysis of $e_\mu$ (and also all the other eigenvalues of $\mathcal{V}_\mu$) by requiring certain definiteness conditions of $V$ (cf.~Lemma~\ref{lem:3} and Lemma~\ref{lem:4}). 
\begin{defi}{\rm (Admissible potentials)} \label{def:admpot} 
Let $V \in L^1(\R^3) \cap L^{3/2}(\R^3)$ be a real--valued radial function and define
\begin{equation} \label{eq:defs}
s^*_{\pm} := \sup \left\{ s \ge 0 :  \vert \cdot \vert^{-s}V_{\pm} \in L^1(\R^3) \right\} \qquad s^* := \min(s^*_+, s^*_-)\,,
\end{equation}
where $V_\pm = \max\{ \pm V, 0\}$ are the positive and negative parts of $V$. We call $V$ an  \emph{admissible potential} if the following is satisfied: 
\begin{itemize}
\item[(a)] There exists $a>0$ such that 
\begin{equation*}
\sup \left\{ r \ge 0 : \lim\limits_{\varepsilon \to 0} \frac{1}{\varepsilon^r} \int_{B_\epsi} \hspace{-0.5mm} V_{\pm}(x) \D x = 0 \right\} = \sup \left\{ r \ge 0 : \lim\limits_{\varepsilon \to 0} \frac{1}{\varepsilon^r} \int_{B_\epsi} \hspace{-0.5mm}  V_{\pm}\vert_{B_a}^* (x) \, \D x = 0 \right\} \,,
\end{equation*}
where $ V_{\pm}\vert_{B_a}^*$ denotes the symmetric decreasing rearrangement of $V_{\pm}\vert_{B_a}$, the restriction of $V_{\pm}$ to the ball of radius $a$ around $0$,
\item[(b)] if $\vert \cdot \vert^{-2}V\notin L^1(\R^2)$, we have $s^* = s^*_- < s^*_+$, if $\vert \cdot \vert^{-2}V \in L^1(\R^2)$, we have $\int_{\R^3} \hspace{-1mm}\frac{V(x)}{\vert x \vert^2} \D x < 0$,
\item[(c)] $s^* > 1$, and
\item[(d)] if $s^* \ge 53/27$, we have $V \in L^p(\R^3)$ for some~${p >5/3}$. 
\end{itemize}
Condition (d) can be dropped, whenever we have control on the ground state space of $\mathcal{V}_\mu$ in the following sense: There exists $\mu_0 >0$ and $\mathcal{L} \subset \mathbb{N}_0$ with $\vert \mathcal{L} \vert < \infty$, such that for all $\mu \ge \mu_0$, the ground state space of $\mathcal{V}_\mu$ is contained in the subspace of $L^2(\Sph^2)$ spanned by the spherical harmonics with angular momentum $\ell \in \mathcal{L}$. 
\end{defi}
\noindent In a nutshell, an admissible potential is a radial potential $V \in L^1(\R^3) \cap L^{3/2}(\R^3)$, which satisfies the following: 
\begin{itemize}
	\item[(i)] There exists some $a>0$ such that both, positive and negative part, have their strongest singularity in $B_a$ at the origin.
	\item[(ii)] It has a dominating attractive part (for short distances), i.e.~$s^*_- \hspace{-1.5mm}< s^*_+$ resp.~${\int_{\R^3} \hspace{-1mm}\frac{V(x)}{\vert x \vert^2} \D x < 0}$. 
\item [(iii)] It is slightly less divergent at the origin than allowed by the $L^{3/2}(\R^3)$-assumption, i.e.~$s^* > 1$.
\end{itemize} 
The most relevant examples for admissible potentials are the attractive Yukawa and Gaussian potential, 
i.e.
\begin{equation*}
V_{\mathrm{Yukawa}}(x) = -\frac{1}{4\pi \vert x \vert}\E^{-\vert x \vert } \qquad \text{and} \qquad V_{\mathrm{Gauss}}(x) = -  (2\pi)^{-3/2} \E^{-\vert x \vert^2 /2} \,.
\end{equation*}
\begin{rmk} \label{rmk:adm} {\rm (On condition (d) for admissible potentials)}\\
The additional $L^{p}(\R^3)$-assumption with $p >5/3$ for  $s^* \ge 53/27$ in condition (d) is due to technical reasons and will we be explained during the proof of Theorem \ref{thm:2}, which is formulated below.  Note that, since $s^* \ge 53/27$ and $V \in L^1(\R^3)$, this condition is essentially about regularity away from $0$ and infinity. However, our proof would work without change if we only had $p > f(s^*)$, where $f$ has some complicated (explicit) expression (see Lemma \ref{lem:3} and Equation \eqref{eq:fsstar}) and is strictly monotonically increasing between $53/27$ and $2$, and satisfies $f(53/27) = 3/2$ and $f(s^*) = 5/3$ for all $s^* \ge 2$. We do not state Theorem \ref{thm:2} with this slight generalization for simplicity. Whenever we have some control on the ground state space of $\mathcal{V}_\mu$, the  $L^p(\R^3)$-assumption is not necessary. For example, in the special case $\hat V \le 0$,  one can choose $\mathcal{L} = \{ 0 \}$ by means of \Cref{eq:emuangmom0} and completely drop condition (d). 
\end{rmk}
\noindent We will show in Lemma~\ref{lem:4} that for any admissible potential $e_\mu <0$ for $\mu$ large enough. Moreover, for any radial potential $V \in L^1(\R^3) \cap L^{3/2}(\R^3)$ with $e_\mu <0$ and $s^*>1$ (in particular any admissible potential), by application of Theorem~\ref{thm:1}, the critical temperature decays exponentially fast as $\mu \to \infty$ since
\begin{equation} \label{elbound}
\vert \sqrt{\mu} e_\mu \vert \le \frac{1}{2 \pi^2} \left\Vert \frac{V}{\vert \cdot \vert^{{s}}} \right\Vert_{L^1} \sup_{\ell \in \mathbb{N}_0} \left\Vert \vert \cdot \vert^{s/2} j_{\ell} \right\Vert_{L^\infty}^2 \mu^{\frac{1-s}{2}}
\end{equation}
for $s \in (1,s^*)$ and the term involving $j_\ell$ is finite as long as $s \le 5/3$ by uniform decay of spherical Bessel functions (see Proposition~\ref{prop:2}~(iii)). A slightly different bound as given in Lemma~\ref{lem:3} allows to improve this threshold. Note that the class of interaction potentials from Proposition~\ref{prop:1} is not admissible since $s^* = 1$ for these potentials. 

The existence of a maximal critical temperature at some intermediate density (\textit{superconducting dome}), can now be obtained by combining the decay of $T_c$ in the high--density limit from \Cref{eq:leadingorder} and \Cref{elbound} for admissible potentials in the sense of Definition~\ref{def:admpot} to the decay of $T_c$ in the low--density limit, where
\begin{equation} \label{eq:lowdensityTc}
T_c =  \mu\left(\frac{8}{\pi} \E^{\gamma - 2} + o(1)\right) \E^{\pi /(2\sqrt{\mu}a)} 
\end{equation}
as shown in \cite{hs08}. This result was obtained for (not necessarily radially symmetric) real valued interaction potentials $V$, with $V(x)(1+\vert x \vert) \in L^1(\R^3) \cap L^{3/2}(\R^3)$, negative scattering length $a$, and in the absence of bound states. Thus, we rigorously confirmed the \textit{ubiquity of superconducting domes in BCS theory} for a general class of interaction potentials, as claimed in \cite{langmann}.

As our next result, we shall derive the second order correction to \Cref{eq:leadingorder}, i.e.~we shall compute the constant in front of the exponential for admissible potentials. For this purpose we define the operator $\mathcal{W}_\mu^{(\kappa)}$ on $u \in L^2(\Sph^2)$ via its quadratic form 
\begin{align} \label{eq:defWmu}
\big\langle u \big\vert\mathcal{W}_\mu^{(\kappa)} \big\vert u \big\rangle = \sqrt{\mu}\int_{0}^{\infty} \D \vert p \vert &\left(\frac{\vert p \vert^2}{\vert \vert p \vert^2 - 1 \vert} \left[\int_{\Sph^{2}}\D \omega(p) \left(\vert \hat{\varphi}(\sqrt{\mu} p) \vert^2 - \vert \hat{\varphi}(\sqrt{\mu} p/\vert p \vert) \vert^2\right)\right]\right. \nonumber \\
&+ \left. \frac{\vert p\vert^2}{\vert p\vert^2 + \kappa^2} \int_{\Sph^{2}} \D \omega(p) \vert \hat{\varphi}(\sqrt{\mu} p/\vert p \vert) \vert^2
\right)
\end{align}
for fixed $\kappa \ge  0$ (cf.~Equation (13) in \cite{hs081} for an analogous definition in the weak coupling case with $\kappa = 0$).
Here, we denote $\hat{\varphi}(p) = (2\pi)^{-3/2} \int_{\Sph^{2}} \hat{V}(p-\sqrt{\mu}q) u(q) \D \omega(q)$, and $(\vert p\vert , \omega(p)) \in (0,\infty) \times \Sph^2$ are spherical coordinates for $p \in \mathbb{R}^3$. Since $V \in L^1(\R^3)$, the map $\vert p \vert \mapsto \int_{\Sph^{2}} \D \omega(p) \vert \hat{\varphi}(p)\vert^2$ is Lipschitz continuous for any $u \in L^2(\Sph^2)$, such that the radial integral in \Cref{eq:defWmu} is well defined even in the vicinity of $\vert p\vert \sim 1$. For large $\vert p \vert$ the integral converges since $V \in L^{3/2}(\R^3)$.  Although we formulate our result in Theorem \ref{thm:2} only for $\kappa = 0$, the case of a positive parameter $\kappa >0$ is crucial in the proof of this statement, as it ensures, e.g., that the first term in the decomposition of the Birman--Schwinger operator associated with $K_{T,\mu} + V$ is small (cf.~\Cref{bsdecomp}). Whenever it does not lead to confusion, we refer to some $\kappa$--dependent quantity at $\kappa = 0$ by simply dropping the  $(\kappa)$--superscript. 

Now, we define the operator
\begin{equation} \label{eq:Bmu}
\mathcal{B}_\mu^{(\kappa)} = \frac{\pi}{2}\left(\mathcal{V}_\mu - \mathcal{W}_\mu^{(\kappa)}\right)\,,
\end{equation}
which measures the strength of the interaction potential near the Fermi surface up to second order and let $b_\mu^{(\kappa)}$ denote its lowest eigenvalue, 
\begin{equation} \label{eq:bmu}
b_\mu^{(\kappa)} = \mathrm{inf \, spec} \, \mathcal{B}_\mu^{(\kappa)}\,.
\end{equation}
We introduced the factor $\pi/2$ in \Cref{eq:Bmu} since $b_\mu = b_\mu^{(0)}$ has the interpretation of an effective scattering length, which is best illustrated in the case of small $\mu$ (see Proposition~1 in \cite{hs081}). Moreover, we will see in the proof of Theorem~\ref{thm:2} that if $e_\mu <0$ then also $b_\mu^{(\kappa)}<0$ for large enough $\mu$. With the aid of $b_\mu^{(\kappa)}$ we can now state our second main result concerning the asymptotic formula for the critical temperature valid up to second order. 
\begin{thm}~\label{thm:2} Let $V$ be an admissible potential. Then the critical temperature $T_c$ is positive and satisfies
	\begin{equation*} 
\lim\limits_{\mu \to \infty} \left(\log\frac{\mu}{T_c} + \frac{\pi}{2\sqrt{\mu} b_\mu}\right) = 2-\gamma -  \log(8/\pi)\,.
	\end{equation*} 
\end{thm}
\noindent  In other words, 
\begin{equation*} 
T_c = \mu \left(\frac{8}{\pi}\E^{\gamma-2} + o(1)\right)\E^{\pi/(2 \sqrt{\mu} b_\mu)}
\end{equation*}
in the limit $\mu \to \infty$. Similarly to Theorem~\ref{thm:1}, this formula is in complete analogy to the weak--coupling case \cite{hs081} (replace $V \to \lambda V$ and take the limit $\lambda \to 0$) but we have coupling parameter $\lambda =1$ here. As discussed in the introduction, this analogy is not entirely surprising. In physical terms, only those fermions with momenta close to the Fermi surface $\{p^2 = \mu \}$ contribute to the superconductivity. Therefore, due to the decay of the interaction $\hat{V}$ in Fourier space, the high--density limit, $\mu \to \infty$, is effectively a weak--coupling limit.

The constant in front of the exponential is in particular relevant for obtaining the universality of the ratio of the critical temperature and the energy gap, which is achieved in \cite{HLenergygap}, where a similar asymptotic formula for the energy gap is proven. 
\begin{rmk}{\rm (Radiality)} \label{rmk:2} \\
The assumption of the interaction potential being radially symmetric enters the proofs of our main theorems in a crucial way. On the one hand, the radial symmetry allows an additional averaging over the sphere $\Sph^2$ in position space in the proof of Theorem~\ref{thm:1}, which leads to a ``decoupling" of the position variables $x$ and $y$ (cf.~\Cref{eq:decoupling}) as the arguments of integral kernels of operators that appear after employing the Birman--Schwinger principle \cite{hhss, fhns, hs15}. Without this averaging the supposed error terms in \Cref{bsdecomp} could not be concluded to be small. On the other hand, the radial symmetry enables us to obtain useful bounds on the quantity $e_\mu$ (cf.~Lemma~\ref{lem:2}, Lemma~\ref{lem:3}, and Lemma~\ref{lem:4}), which naturally appears in the obtained asymptotics in Theorem~\ref{thm:1} and Theorem~\ref{thm:2}. Although the assumption of a radial potential is a loss of generality compared to the weak coupling \cite{fhns, hs081} and low density\cite{hs08} situation, the case of an isotropic interactions seems physically the most relevant and natural. 
\end{rmk}
\begin{rmk} {\rm (Potentials with slow decay at infinity)} \label{rmk:3} \\
The recent work \cite{cueninmerz} by Cuenin and Merz indicates how to generalize our results to interaction potentials with slow decay at infinity, i.e.~which fail to satisfy $V \in L^1(\R^3)$. The main idea is to employ the Tomas--Stein Theorem~to define the Fourier transform of the potential on the $\mathrm{codim-}1$ submanifold $\Sph^2 \subset \R^3$ having non--vanishing curvature. Moreover, by using the methods  from \cite{gontier}, where Gontier, Hainzl, and Lewin originally studied a lower bound on the Hartree--Fock energy of the electron gas, one can see that  
\begin{equation} \label{eq:gontierbound}
T_c \le C_1\,  \mu \exp\left(- C_2\,  \mu^{1/4}\right) 
\end{equation}
for any real--valued potential $V$ satisfying $\vert\cdot\vert V \in L^{\infty}(\R^3)$. A detailed proof of this estimate is given in Section \ref{sec:proofs}. Note that for an admissible potential $V$ that satisfies $\vert \cdot \vert V \in L^\infty(\R^3)$, we have $s^* \ge 2$ and infer by \Cref{elbound} that the bound provided by \Cref{eq:gontierbound} is not optimal. Although these results indicate that it is mathematically possible to deal with slow decay at infinity, it seems physically natural to assume fast decay at infinity, at least in the high--density limit for an effective interaction potential, when the phenomenon of screening plays an important~role.
\end{rmk}

\section{Proofs} \label{sec:proofs}
The most important tool for our proofs will be the Birman--Schwinger principle (see \cite{hhss, fhns, hs15}). According to this principle, $T_c$ is determined by the fact that for $T= T_c$ the smallest eigenvalue of 
\begin{equation*}
B_{T,\mu} = V^{1/2} \frac{1}{K_{T,\mu}} \vert V \vert^{1/2}
\end{equation*}
equals $-1$. Here, we used the notation $V(x)^{1/2} = \mathrm{sgn}(V(x)) \vert V(x)\vert^{1/2}$. The main simplification is that the study of the spectrum of the unbounded operator $K_{T,\mu} + V$ reduces to identifying the singular part of the compact Birman--Schwinger operator. With this in mind, our proofs will build on a convenient decomposition of $B_{T,\mu}$ in a dominant singular term and other error terms. 

Let ${\mathfrak{F}_{\mu}} : L^1(\R^3) \to L^2(\Sph^{2})$ denote the rescaled Fourier transform restricted to $\Sph^{2}$ with
\begin{equation*}
\left(\mathfrak{F}_{\mu} \psi \right)(p) = \frac{1}{(2\pi)^{3/2}} \int_{\R^3} \E^{-\mathrm{i} \sqrt{\mu} p\cdot x} \psi(x) \D x\,,
\end{equation*}
which is well--defined by the Riemann--Lebesgue Lemma.
Since $V \in L^1(\R^3)$, the multiplication with $\vert V \vert^{1/2}$ is a bounded operator from $L^2(\R^3)$ to $L^1(\R^3)$, and hence $\mathfrak{F}_{\mu} \vert V \vert^{1/2}$ is a bounded operator from $L^2(\R^3)$ to $L^2(\Sph^{2})$. A further important ingredient in our proofs is to study the asymptotic behavior of 
\begin{equation*}
m_{\mu}^{(\kappa)}(T) = \frac{1}{4\pi} \int_{\R^3} \left(\frac{1}{K_{T,\mu}(p)} - \frac{1}{p^2+\kappa^2\mu}\right)\D p
\end{equation*}
for fixed $\kappa >0$, which was done in a similar way for the low--density and weak--coupling limit of the critical temperature and the energy gap (see \cite{hs081, hs08, hs15, lauritsen}). Indeed, using Lemma~1 from \cite{hs081} one can easily see that
\begin{equation} \label{masymp}
m_{\mu}^{(\kappa)}(T) = \sqrt{\mu} \left( \log \frac{\mu}{T} + \gamma - 2+ \log\frac{8}{\pi} + \kappa \, \frac{\pi}{2} + o(1)\right)
\end{equation}
as long as $T/\mu \to 0$.
Using the definitions above, we arrive at our convenient decomposition, which we define as
\begin{equation} \label{bsdecomp}
B_{T,\mu} = V^{1/2} \frac{1}{p^2 + \kappa^2 \mu} \vert V \vert^{1/2} + m_\mu^{(\kappa)}(T) \, V^{1/2} {\mathfrak{F}_{\mu}}^\dagger{\mathfrak{F}_{\mu}} \vert V \vert^{1/2} + A^{(\kappa)}_{T,\mu}
\end{equation}
for $\kappa > 0$, where $A_{T,\mu}^{(\kappa)}$ is implicitly defined. For the first term to be small, we need that $\kappa > 0$. For the second term, note that 
\begin{equation*}
V^{1/2} {\mathfrak{F}_{\mu}} ^\dagger{\mathfrak{F}_{\mu}} \vert V \vert^{1/2}
\end{equation*}
is isospectral to $ \mathcal{V}_\mu = {\mathfrak{F}_{\mu}} V {\mathfrak{F}_{\mu}}^\dagger$.   In fact, the spectra agree at first except possibly at $0$, but $0$ is in both spectra as the operators are compact on an infinite dimensional space. 

This second term will be the dominant term, which is how the quantity $e_\mu$ appears in the asymptotic formulas in our main theorems, whereas the first and third term are negligible error terms in the limit $\mu \to \infty$. Showing this, is the objective of the proofs of Theorem~\ref{thm:1} and Theorem~\ref{thm:2}

\textit{A priori}, it is not clear, how $T_c$ behaves at high densities. Therefore, before we go to the proofs of Theorem~\ref{thm:1} and Theorem~\ref{thm:2}, let us fix the following 
\begin{lem} \label{lem:1} 
$T_c \le O(\mu)$ as $\mu \to \infty$. 
\end{lem}
\begin{proof}
Since $\tanh(t) \le \min(1,t)$ for $t \ge 0$, we have
\begin{align*}
K_{T,\mu} + V &\ge \frac{1}{2} \left(\vert p^2 - \mu \vert +2T\right) + V \\
& \ge \frac{1}{2} \left( p^2 + \mu  +2V\right) + (T-\mu)\,.
\end{align*}
The first term is non--negative for sufficiently large $\mu$ by application of Sobolev's inequality \cite[Thm.~8.3]{liebloss} using $V\in L^{3/2}(\R^3)$. Thus, by Definition \ref{def:1}, we obtain $T_c \le \mu$. 
\end{proof}
\noindent In the proof of Theorem~\ref{thm:1} below, we will in fact show that $T_c = o(\mu)$, so  \Cref{masymp} gives the correct asymptotic behavior.

\subsection{Proof of Theorem~\ref{thm:1}}
\begin{proof}[Proof of Theorem~\ref{thm:1}]
Fix $\kappa >0$. As outlined above, the strategy of the proof is to show that the first and the third term in the decomposition \eqref{bsdecomp} vanish in operator norm in the high--density limit and thus the asymptotic behavior is entirely determined by the spectrum of the operator in the second term. We discuss this in detail now. 

 For the first term, note that the Fourier transform of $\frac{1}{p^2 + \kappa^2 \mu}$ is given by $\frac{\E^{- \kappa \sqrt{\mu} \vert x \vert}}{\vert x \vert}$, up to a constant. Thus the Hilbert--Schmidt norm $\Vert \cdot \Vert_{\mathrm{HS}}$, which is always an upper bound for the operator norm $\Vert \cdot \Vert_{\mathrm{op}}$, is given by
\begin{equation*}
\left\Vert V^{1/2} \frac{1}{p^2 + \kappa ^2 \mu} \vert V\vert^{1/2} \right\Vert_{\mathrm{HS}}^2 = C \int_{\R^3}\D x \int_{\R^3}\D y \, \vert V(x) \vert \frac{\E^{- 2 \kappa \sqrt{\mu} \vert x-y \vert}}{\vert x -y\vert^2}\vert V(y) \vert\,.
\end{equation*}
This vanishes as $\mu \to \infty$ by an application of the dominated convergence theorem in combination with the Hardy--Littlewood--Sobolev inequality \cite[Thm.~4.3]{liebloss}. Here and in the following, we shall use the notation $c$ and $C$ for generic positive small resp.~large constants, possible having a different value in each appearance. If we want to emphasize the dependence on a certain parameter, we add a subscript, e.g.~by writing $c_\delta$ or $C_{p,a}$.

For the third term, we will heavily use the radiality of $V$. In fact, since $V$ is radially symmetric, every eigenfunction of $K_{T,\mu}$ and thus $B_{T}$ will have definite angular momentum and we can focus on $f,g \in L^2(\R^3)$ of the form $f(x) = f(\vert x \vert) Y_\ell^m(\hat{x})$ resp.~$g(x) = g(\vert x \vert) Y_{\ell'}^{m'}(\hat{x})$, with a slight abuse of notation, where $\hat{x} = x/\vert x \vert$ denotes the unit vector in direction $x$. Now we aim to bound $\langle f\vert  A^{(\kappa)}_{T,\mu}\vert g\rangle$, uniformly in $(\ell, \ell')$ (and $(m,m')$). As $A_{T,\mu}^{(\kappa)}$ has integral kernel
\begin{equation*}
A_{T,\mu}^{(\kappa)}(x,y) = C V^{1/2}(x) \vert V(y) \vert^{1/2} \int_{\R^3} \left(\frac{1}{K_{T,\mu}(p)}- \frac{1}{p^2 + \kappa^2 \mu}\right) \left(\E^{\I p\cdot (x-y)}- \E^{\I \sqrt{\mu}\hat{p}\cdot (x-y)}\right) \D p\,,
\end{equation*}
and using the radial symmetry of $V$ we arrive at 
\begin{align}
 \big\langle f\big\vert  A^{(\kappa)}_{T,\mu}\big\vert g\big\rangle\, = \ &C \int_{0}^{\infty} \D \vert x\vert\,  |x|^2 \int_{0}^{\infty} \D \vert y\vert \, |y|^2 \bar{f}(\vert x\vert) V^{1/2}(\vert x\vert) \vert V(\vert y\vert) \vert^{1/2} g(|y|) \label{eq:firstline}\\
&\times \ \int_{\R^3} \D p \, \left(\frac{1}{K_{T,\mu}(p)}- \frac{1}{p^2 + \kappa^2 \mu}\right)  \label{eq:secondline}\\
& \times \ \int_{\Sph^{2}} \D\omega(x) \int_{\Sph^{2}}\D\omega(y) \overline{Y}_\ell^m(\hat{x}) Y_{\ell'}^{m'}(\hat{y})\left(\E^{\I p\cdot (x-y)}- \E^{\I \sqrt{\mu}\hat{p}\cdot (x-y)}\right) \,.\label{eq:thirdline}
\end{align}
Note that the angular integrals of $x$ and $y$ can be performed first only by the radial symmetry of $V$. If $V$ were not radially symmetric, we would have had to compute the angular integral of $p$ first and would have ended up with completely different integrals to estimate.
Now, using the plane wave expansion $\E^{\I p\cdot x} = 4\pi \sum_{\ell=0}^{\infty} \sum_{m= -\ell}^{\ell} \I^\ell j_\ell(|p||x|) Y_{\ell}^m(\hat p) \overline{Y}_\ell^m(\hat x)$, the last line~\eqref{eq:thirdline} evaluates to 
\begin{equation} \label{eq:decoupling}
16\pi^2 \, (-\I)^{\ell + \ell'}\left(j_{\ell}(\vert p \vert \vert x \vert ) j_{\ell'}(\vert p \vert \vert y \vert ) - j_{\ell}(\sqrt{\mu} \vert x \vert ) j_{\ell'}(\sqrt{\mu} \vert y \vert )\right)\overline{Y}_\ell^m(\hat{p}) Y_{\ell'}^{m'}(\hat{p}) \,.
\end{equation}
In order to get a non-zero angular $p$-integral from the second line \eqref{eq:secondline}, we need to have $\ell = \ell'$ and $m=m'$, by orthogonality of spherical harmonics. We can hence focus on that case and write $x$, $y$, and $p$ instead of $\vert x \vert$, $\vert y \vert $, and $\vert p \vert$, respectively, such that \eqref{eq:secondline} and \eqref{eq:thirdline} combine to (a constant times)
\begin{equation}
\int_{0}^{\infty} \D p \, p^2 \left(\frac{1}{K_{T,\mu}(p)}- \frac{1}{p^2 + \kappa^2 \mu}\right)\left(j_\ell(p x ) j_\ell(p y ) - j_\ell(\sqrt{\mu} x ) j_\ell(\sqrt{\mu} y )\right) \,. \label{eq:CTmu}
\end{equation}
After changing the integration variable $p \to p/\sqrt{\mu}$ and inserting $\pm j_\ell(p \sqrt{\mu} x ) j_\ell(\sqrt{\mu} y )$, we use the uniform Lipschitz continuity and the uniform decay of spherical Bessel functions (Proposition \ref{prop:2} (ii) and Proposition \ref{prop:2} (iii)) to obtain
\begin{align*}
\left\vert \eqref{eq:CTmu} \right\vert \le C \mu^{1/2}\hspace{-2mm} &\int_{0}^{\infty}   \D p \, p^2 \left\vert \frac{1}{K_{T/\mu,1}(p)}- \frac{1}{p^2 + \kappa^2}\right\vert \vert p-1\vert^{\epsi} \left(\frac{1}{p^\epsi}+ 1\right)  \\
\times & \left(\vert j_\ell(p\sqrt{\mu}x) \vert^{1-11 \epsi/5} + \vert j_\ell(\sqrt{\mu}x) \vert^{1-11 \epsi/5} \right) \hspace{-1mm}\left(\vert j_\ell(p\sqrt{\mu}y) \vert^{1-11 \epsi/5} + \vert j_\ell(\sqrt{\mu}y) \vert^{1-11 \epsi/5} \right)\,,
\end{align*}
for any $\epsi \in (0,2/11)$. Using that and by employing Hölder for the integrals over $x$ and $y$ in \Cref{eq:firstline}, we get
\begin{align}
\big\vert \big\langle f\big\vert  A^{(\kappa)}_{T,\mu}\big\vert g\big\rangle\big\vert \ \le \  &C \,  \Vert f\Vert_{L^2}  \, \Vert g\Vert_{L^2}   \int_{0}^{\infty} \D p \, \left\vert \frac{1}{K_{T/\mu,1}(p)}- \frac{1}{p^2 + \kappa^2 }\right\vert \vert p-1\vert^{\epsi} \left(\frac{1}{p^\epsi}+ 1\right)p^2\\
&\times \,\mu^{1/2}  \int_{\R^3} \D x \,  \vert V (x)\vert \left(\vert j_\ell(p\sqrt{\mu}\vert x \vert) \vert^{2-22\epsi/5} + \vert j_\ell(\sqrt{\mu}\vert x \vert) \vert^{2-22\epsi/5}\right)  \,. \label{eq:omubound}
\end{align}
 In Lemma~\ref{lem:2} below (as $\epsi <2/11$ we have $2-22\epsi/5 >6/5$), we show that the last line \eqref{eq:omubound} can be estimated by 
\begin{equation*}
\eqref{eq:omubound} \le  \left(\frac{1}{p} + 1\right) \vert o(1)\vert \,,
\end{equation*}
where $o(1)$ is some function of $\mu$ that vanishes as $\mu \to \infty$. Thus, we arrive at 
\begin{align*}
\big\vert \big\langle f\big\vert  A^{(\kappa)}_{T,\mu}\big\vert g\big\rangle\big\vert \ \le \  &C \, \vert o(1) \vert \,   \Vert f\Vert_{L^2}\, \Vert g\Vert_{L^2}    \int_{0}^{\infty} \D p \,  \left\vert \frac{1}{K_{T/\mu,1}(p)}- \frac{1}{p^2 + \kappa^2}\right\vert \vert p-1\vert^{\epsi} p^{1-\epsi}(1+p^{1+\epsi}) \,,
\end{align*}
where the integral is uniformly bounded (since $\kappa >0$)  as long as $T \le C \mu$ and we conclude 
\begin{equation*}
\limsup\limits_{\mu \to \infty} \sup_{0<T\le C \mu} \left\Vert A_{T,\mu}^{(\kappa)} \right\Vert_{\mathrm{op}} = 0\,,
\end{equation*}
since all bounds above are uniform in $\ell$. Therefore, as long as $T_c = O(\mu)$, the spectrum of the Birman--Schwinger operator approaches the spectrum of the operator in the second term in \Cref{bsdecomp} as $\mu \to \infty$. 

Summarizing our considerations above, we get that, since by assumption $e_\mu< 0$ for $\mu \ge \mu_0$, $T_c>0$ for all sufficiently large $\mu$. This is because the second term in the decomposition~\eqref{bsdecomp} can be made arbitrarily negative by taking $T \to 0$, whereas the first and the third term are bounded uniformly in $T \le C \mu$. Thus we get with the aid of Lemma~\ref{lem:1} that 
\begin{equation*}
- 1 = \lim\limits_{\mu \to \infty} m_\mu^{(\kappa)} (T_c) \, e_\mu\,.
\end{equation*}
In order to obtain \Cref{eq:thm1} by means of \Cref{masymp}, it remains to show that $T_c = o(\mu)$. Since it is already shown in Lemma~\ref{lem:1} that $T_c = O(\mu)$, we assume that $T_c = \Theta(\mu)$, i.e.~there exist $0<c<C$ such that $c \mu \le T_c \le C \mu$. This means that $m_{\mu}^{(\kappa)}(T_c)$ is of order $\sqrt{\mu}$, which leads to a contradiction since $\sqrt{\mu} e_\mu= o(1)$ by Lemma~\ref{lem:2} below. So, \Cref{masymp} implies \Cref{eq:thm1} as desired. 
\end{proof}
\begin{lem} \label{lem:2} Let $V \in L^{3/2}(\R^3)$ and $\alpha >6/5$. Then 
	\begin{equation*}
\limsup \limits_{\mu \to \infty} \sqrt{\mu} \sup_{\ell \in \mathbb{N}_0}\int_{\R^3}\D x \, \vert V(x)\vert \, \left\vert  j_\ell(\sqrt{\mu} \vert x \vert)\right\vert^\alpha = 0\,.
	\end{equation*}
\end{lem}
\begin{proof}
We estimate
\begin{align}
\sqrt{\mu} &\sup_{\ell \in \mathbb{N}_0}\int_{\R^3}\D x \, \vert V(x)\vert \, \left\vert j_\ell(\sqrt{\mu} \vert x \vert)\right\vert^\alpha \, \le C \, \sqrt{\mu} \int_{\R^3}\D x \, \vert V(x)\vert  \frac{1}{\left(\sqrt{\mu} \vert x \vert \right)^{5\alpha/6}+1} \label{eq:limsup} \,,
\end{align}
where the inequality follows from the uniform decay of spherical Bessel functions (see Proposition~\ref{prop:2} (iii)). By using Hölder, we can further bound 
\begin{equation*}
\eqref{eq:limsup} \le C \Vert V- \phi \Vert_{L^{3/2}} \left\Vert \frac{1}{\vert \cdot \vert^{5\alpha/6} + 1} \right\Vert_{L^3} + C \, \sqrt{\mu} \int_{\R^3}\D x \, \vert \phi(x)\vert  \frac{1}{\left(\sqrt{\mu} \vert x \vert \right)^{5\alpha/6}+1}
\end{equation*}
for any $\phi \in C_0^\infty(\R^3)$. Since $\alpha>6/5$, the second term vanishes as $\mu \to \infty$ by dominated convergence using $\phi \in C_0^\infty(\R^3)$, and the first term can be made arbitrarily small as $C_0^\infty(\R^3)$ is dense in $L^{3/2}(\R^3)$. Thus, we have proven the claim. 
\end{proof}

\subsection{Proof of Proposition~\ref{prop:1}}
\begin{proof}[Proof of Proposition~\ref{prop:1}]
	We check that $V_\alpha$ satisfies the assumptions of Theorem \ref{thm:1}. First, $V_\alpha $ is radial and clearly satisfies $V_\alpha \in L^1(\R^3)$. $V_\alpha \in L^{3/2}(\R^3)$ follows since $\alpha >1/3$.  Next, we calculate the derivative of $\vert x \vert \vert V_\alpha (x) \vert $ w.r.t.~$\vert x \vert$ as
	\begin{equation*}
\left(\frac{(\log(\vert x \vert) + \alpha)^2}{\vert x \vert (\log^2(\vert x \vert) + 1)} + \frac{1-\alpha^2}{\vert x \vert (\log^2(\vert x \vert) + 1)} + 1\right) \frac{\exp(-\vert x \vert )}{\vert x\vert  (\log^2(\vert x \vert) + 1)}
	\end{equation*}
	and conclude that $\vert x \vert V_\alpha (x) $ is monotonically increasing in $\vert x \vert$, since $\alpha < 1/2$. Therefore, by using the radiality of $V_\alpha$ and the argument from Equation (4) in \cite{Fourier}, we find that $\hat{V}_\alpha \le 0$ and infer
\begin{equation*}
e_\mu = \frac{1}{2\pi^2} \int_{\R^3} V_\alpha(x) \left(\frac{\sin(\sqrt{\mu} \vert x \vert)}{\sqrt{\mu} \vert x \vert}\right)^2 \D x < 0\,
\end{equation*}
by \Cref{eq:emuangmom0}. Thus, $V_\alpha$ satisfies all conditions of Theorem \ref{thm:1}. In order to obtain a lower bound on $T_c(\mu)$ we estimate
\begin{equation*}
\vert \sqrt{\mu} e_\mu \vert \ge c \int_{0}^{\sqrt{\mu}/2} \frac{\sin(x)^2}{x^2 \vert \log(x/\sqrt{\mu})\vert^{2\alpha}}\D x \ge c\,  \frac{1}{\vert \log(\mu)\vert^{2\alpha}}
\end{equation*}
for some $c>0$ and $\mu$ large enough. Using Theorem~\ref{thm:1}, this yields
\begin{equation*}
T_c \gtrsim \mu \exp(-\log(\mu)^{2\alpha}/c) \to  \infty
\end{equation*}
as $\mu \to \infty$ since $\alpha < 1/2$. 
\end{proof}

\subsection{Proof of Theorem~\ref{thm:2}} \label{subsec:thm2}
The proof of Theorem~\ref{thm:2} is based on the following two Lemmas providing the necessary estimates for a perturbation theoretic argument yielding the next order correction to the asymptotics obtained in Theorem~\ref{thm:1}.  While Lemma~\ref{lem:3} improves the upper bounds on integrals of the interaction potential against spherical Bessel functions obtained in Lemma~\ref{lem:2} and \Cref{elbound}, in particular for $s^*> 5/3$, Lemma~\ref{lem:4} provides a lower bound on $e_\mu$ for admissible potentials. We postpone the proofs of Lemma \ref{lem:3} and Lemma \ref{lem:4} until the end of this Section. 
\begin{lem} \label{lem:3}
Let $V \in L^1(\R^3) \cap L^p(\R^3)$ for some $p \in [3/2,9/4]$ with dual $q = \frac{p}{p-1} \in [9/5,3]$. Assume that $s^* >1$, with $s^*$ as in Definition \ref{def:admpot} and set
\begin{equation} \label{eq:betabound}
\beta^*_p = \begin{cases}\frac{s^*}{2}  &\text{for} \quad s^* \in (1,5/3] \\
\min\left( \frac{(q+1)s^*-4}{3qs^*-7}+\frac{1}{2}, \, \frac{10q-11}{12q-14}\right) \quad &\text{for} \quad s^* >5/3 \,.
\end{cases}
\end{equation}
 Note that $\beta_p^*$ depends continuously on $s^*$ and is (strictly) monotonically increasing (between $1$ and $2$), and $\beta^*_p \le \min(s^*,2)/2$ for any $s^*>1$. Then for any $\delta >0$ there exists an $\epsi_0 >0$ such that for all $\epsi \in [0,\epsi_0]$ we have
\begin{equation*}
\limsup\limits_{\mu \to \infty} \mu^{\beta^*_p - \delta} \sup_{\ell \in \mathbb{N}_0}\int_{\R^3} \D x \vert V (x)\vert \, \vert j_\ell(\sqrt{\mu}\vert x \vert )\vert^{2-\varepsilon} = 0\,.
\end{equation*}
\end{lem}  
\noindent For admissible potentials that satisfy the $L^p(\R^{3})$-assumption in condition (d) from Definition~\ref{def:admpot}, we will use this Lemma with 
\begin{equation*} 
\beta^* = \begin{cases}\frac{s^*}{2}  &\text{for} \quad s^* \in (1,5/3] \\
\frac{4s^*-4}{9s^*-7}+\frac{1}{2} \quad &\text{for} \quad s^* \in (5/3,53/27) \\ 
\min\left( \frac{7s^*-8}{15s^*-14}+\frac{1}{2}, \, \frac{7}{8}\right) + \delta_p \quad & \text{for} \quad s^* \ge 53/27\,,
\end{cases}
\end{equation*}
for some  $\delta_p > 0$ since $p > 5/3$. For our perturbation theoretic argument to work in the general case, where we have no control on the ground state space of $\mathcal{V}_\mu$, it turns out to be necessary that
\begin{equation} \label{eq:perttheoryestimate}
4\beta^* - \frac{3 \min(s^*,2)}{2} - \frac{1}{2} > 0 \,,
\end{equation}
which is why we need the $L^p(\R^{3})$-assumption in Definition \ref{def:admpot} for $s^* \ge 53/27$.
The function $f(s^*)$ from Remark \ref{rmk:adm} can explicitly be determined by requiring that 
\begin{equation} \label{eq:fsstar}
4\beta^*_{f(s^*)} - \frac{3 \min(s^*,2)}{2} - \frac{1}{2} =  0 \,.
\end{equation}
\begin{lem} \label{lem:4}
Let $V$ be an admissible potential (cf.~Definition \ref{def:admpot}, condition (d) can be dropped).  Then for any $\delta >0$ there exists $c_\delta>0$ such that 
\begin{equation*}
\liminf\limits_{\mu \to \infty} \vert \mu^{\min(s^*+ \delta ,2)/2} \, e_\mu \vert \ge c_\delta \,.
\end{equation*} 
In particular, for admissible $V$, we have $e_\mu <0$ for $\mu$ large enough. 
\end{lem}
\noindent The proof of  Theorem~\ref{thm:2} is divided in two parts. In the first part, Lemma \ref{lem:5}, we provide an asymptotic formula for $T_c$ with parameter $\kappa >0$. In the second part, Lemma \ref{lem:6}, we asymptotically compare $1/(\sqrt{\mu} b_\mu^{(\kappa)})$ with $1/(\sqrt{\mu} b_\mu)$. By combining these formulas, we obtain Theorem \ref{thm:2}. 
\begin{lem} \label{lem:5}
Let $V$ be an admissible potential and fix $\kappa > 0$. Then the critical temperature $T_c$ is positive and satisfies
\begin{equation} \label{eq:thm2kappa}
\lim\limits_{\mu \to \infty} \left(\log\frac{\mu}{T_c} + \frac{\pi}{2\sqrt{\mu} b_\mu^{(\kappa)}}\right) = 2-\gamma - \kappa \, \frac{\pi}{2} - \log(8/\pi)\,.
\end{equation} 
\end{lem}
\begin{lem} \label{lem:6} Let $V$ be an admissible potential and $\kappa > 0$. Then
	\begin{equation} \label{eq:kappa0}
\lim\limits_{\mu \to \infty} \left(  \frac{\pi}{2\sqrt{\mu} b_\mu} - \frac{\pi}{2\sqrt{\mu} b_\mu^{(\kappa)}}\right) =  \kappa \, \frac{\pi}{2} \,.
	\end{equation}
\end{lem}
\begin{proof}[Proof of Theorem \ref{thm:2}]
By combining Lemma \ref{lem:5} with Lemma \ref{lem:6} we obtain
\begin{align*}
\lim\limits_{\mu \to \infty} \left(\log\frac{\mu}{T_c} + \frac{\pi}{2\sqrt{\mu} b_\mu}\right) &= \lim\limits_{\mu \to \infty} \left(\log\frac{\mu}{T_c} + \frac{\pi}{2\sqrt{\mu} b_\mu^{(\kappa)}}\right) + \lim\limits_{\mu \to \infty} \left(  \frac{\pi}{2\sqrt{\mu} b_\mu} - \frac{\pi}{2\sqrt{\mu} b_\mu^{(\kappa)}}\right) \\[1.5mm]
&= 2-\gamma - \kappa \, \frac{\pi}{2} - \log(8/\pi) + \kappa \, \frac{\pi}{2} = 2-\gamma  - \log(8/\pi)\,. \qedhere
\end{align*} 
\end{proof}
\noindent The rest of this Section is devoted to the proofs of the four Lemmas above. We begin with Lemma \ref{lem:5} and  Lemma \ref{lem:6}. 
\begin{proof}[Proof of Lemma~\ref{lem:5}]
Fix $\kappa > 0$. We first assume condition (d) from Definition \ref{def:admpot} and discuss the changes for the special case afterwards. Similarly to the proof of Theorem~\ref{thm:1}, we show that the first and the third term in the decomposition \eqref{bsdecomp} vanish in operator norm. 

For the first term, we need to improve the estimate from Theorem~\ref{thm:1}, where we employed the easily accessible Hilbert--Schmidt norm as an upper bound to the operator norm. Indeed, using the radial symmetry of $V$, similarly to the bound of $A_{T,\mu}^{(\kappa)}$ in \Cref{eq:omubound},
the operator norm of the compact operator $L_\mu^{(\kappa)} := V^{1/2} (p^2 + \kappa ^2 \mu)^{-1} \vert V\vert^{1/2}$
can be estimated as
\begin{equation} \label{eq:Lbound}
\left\Vert  L_{\mu}^{(\kappa)}\right\Vert_{\mathrm{op}} \ \le C \, \mu^{1/2} \,  \int_{0}^{\infty} \D p \frac{p^2}{p^2 + \kappa^2} \sup_{\ell \in \mathbb{N}_0} \int_{\R^3} \D x \vert V(x)\vert \, \vert j_\ell(\sqrt{\mu} p \vert x \vert ) \vert^2\,,
\end{equation}
which is bounded by $\mu^{-\beta^* + 1/2 + \delta}$ for any $\delta > 0$ by means of Lemma~\ref{lem:3} (note that the $p$--integral is finite since $s^* >1$). Recall from the prove of Theorem \ref{thm:1} (in particular Equation~\ref{eq:omubound}) that 
\begin{align} \label{eq:Abound}
\left\Vert  A_{T,\mu}^{(\kappa)}\right\Vert_{\mathrm{op}} \ \le \  &C \,     \int_{0}^{\infty} \D p \, \left\vert \frac{1}{K_{T/\mu,1}(p)}- \frac{1}{p^2 + \kappa^2 }\right\vert \vert p-1\vert^{\epsi} \left(\frac{1}{p^\epsi}+ 1\right)p^2\\
&\nonumber\times \,\mu^{1/2}  \sup_{\ell \in \N_0}\int_{\R^3} \D x \,  \vert V (x)\vert \left(\vert j_\ell(p\sqrt{\mu}\vert x \vert) \vert^{2-22\epsi/5} + \vert j_\ell(\sqrt{\mu}\vert x \vert) \vert^{2-22\epsi/5}\right)  \,. 
\end{align}
for any $\epsi \in (0,5/11)$. Again by Lemma~\ref{lem:3} we may bound the $x$--integral by $\mu^{-\beta^* + \delta}(1 + p^{-\beta^* + \delta})$ for any $\delta > 0$ and the $p$--integral is finite as long as $0<T \le C \mu$.
We now define, analogously to Equation (28) in \cite{hs081}, 
\begin{equation*}
V^{1/2}M_{T,\mu}^{(\kappa)} \vert V \vert^{1/2} := V^{1/2}K_{T,\mu}^{-1}\vert V \vert^{1/2} - m_\mu^{(\kappa)}(T) V^{1/2}\mathfrak{F}^\dagger_\mu \mathfrak{F}_\mu \vert V \vert^{1/2} = L_\mu^{(\kappa)} + A_{T,\mu}^{(\kappa)}
\end{equation*}
and combine the bounds \eqref{eq:Lbound} and \eqref{eq:Abound} from above to obtain 
\begin{equation} \label{eq:case12}
\limsup\limits_{\mu \to \infty} \mu^{\beta^*-1/2-\delta}\sup_{0<T\le C \mu} \left\Vert V^{1/2} M_{T,\mu}^{(\kappa)} \vert V \vert^{1/2} \right\Vert_{\mathrm{op}} =  0
\end{equation}
for any $\delta >0 $. Also, since
$V^{1/2} \mathfrak{F}_\mu^\dagger\mathfrak{F}_\mu \vert V \vert^{1/2}$ 
is isospectral to $\mathcal{V}_\mu$, so its eigenvalues are given by \Cref{eq:eigenvalues}, one can easily see, using Lemma \ref{lem:3} again, that
\begin{equation} \label{eq:case11}
\limsup\limits_{\mu \to \infty} \mu^{\beta^*-\delta} \left\Vert V^{1/2} \mathfrak{F}_\mu^\dagger
\mathfrak{F}_\mu \vert V \vert^{1/2}\right\Vert_{\mathrm{op}} = 0
\end{equation}
for any $\delta >0$. In particular, since $s^*>1$, the bound \eqref{eq:case12} implies that $1+V^{1/2} M_{T,\mu}^{(\kappa)} \vert V \vert^{1/2}$ is invertible for any $0<T \le C \mu$ and $\mu$ large enough. 

We can thus, following the argument around Equation (30) in~\cite{hs081}, conclude that 
the Birman-Schwinger operator $B_{T,\mu}$ having an eigenvalue $-1$ is equivalent to 
 the selfadjoint operator 
 \begin{equation} \label{eq:bsmin2}
\mathfrak{F}_\mu \vert V \vert^{1/2}\frac{m_\mu^{(\kappa)}(T)}{1+ V^{1/2} M_{T,\mu}^{(\kappa)}\vert V \vert^{1/2} } V^{1/2} \mathfrak{F}_\mu^\dagger\,.
 \end{equation}
 acting on $L^2(\Sph^2)$ having an eigenvalue $-1$.
At $T= T_c$, $-1$ is the smallest eigenvalue of $B_{T,\mu}$, hence \eqref{eq:bsmin2} has an eigenvalue $-1$ for this value of $T$. By continuity and monotonicity of $m_\mu^{(\kappa)}(T)$ one can in fact show that $-1$ is the smallest eigenvalue of \eqref{eq:bsmin2} in this case (cf.~the discussion below Equation~(31) in \cite{hs081}). 

Since $\mathfrak{F}_\mu V \mathfrak{F}_\mu^\dagger = \mathcal{V}_{\mu}$ (see \Cref{eq:vmu}) and $e_\mu = \mathrm{inf\, spec}\, \mathcal{V}_\mu < 0$ by Lemma~\ref{lem:4}, it immediately follows  that 
 \begin{equation*}
-1 = \lim\limits_{\mu\to \infty} \mathrm{inf \, spec} \, \mathcal{V}_\mu \,  m_\mu^{(\kappa)}(T_c) = \lim\limits_{\mu \to \infty} e_\mu m_\mu^{(\kappa)}(T_c)\,,
 \end{equation*}
 which, in combination with the asymptotics \eqref{masymp} and the argument for $T_c = o(\mu)$ from the proof of Theorem~\ref{thm:1}, reproves \eqref{eq:thm1} resp. \eqref{eq:leadingorder}, i.e.~the asymptotic formula for $T_c$ to leading order. To obtain the next order, we employ first order perturbation theory as in the proof of Theorem~1 in~\cite{hs081} (in particular, see Equation (32)).

Indeed, using \Cref{eq:case12} and \Cref{eq:case11}, we can expand the geometric series in \Cref{eq:bsmin2} to first order and employ first order perturbation theory to arrive at 
 \begin{equation} \label{eq:impleq1}
\frac{1}{\sqrt{\mu}} m_\mu^{(\kappa)}(T_c) = \frac{-1}{\sqrt{\mu} e_\mu - \sqrt{\mu} \big\langle u \big\vert \mathfrak{F}_\mu V M_{T_c,\mu}^{(\kappa)} V \mathfrak{F}_\mu^\dagger \big\vert u \big\rangle + O(\mu^{-4\beta^*+\min(s^*,2)/2 + 3/2+\delta})}\,,
\end{equation}
for any $\delta >0$. Here, $u$ is the normalized eigenfunction corresponding to the lowest negative eigenvalue $e_\mu$ of the compact operator $\mathcal{V}_\mu = \mathfrak{F}_\mu V \mathfrak{F}_\mu^\dagger$ (see Lemma \ref{lem:4}). In case of (finite!) degeneracy, one has to choose the ground state $u$ of $\mathcal{V}_\mu$ that minimizes the second term in the denominator of \eqref{eq:impleq1}.
The error term in \Cref{eq:impleq1} is twofold. 
The first part comes from expanding the geometric series. The second part comes from first order perturbation theory, where we made use of the fact that 
\begin{equation} \label{eq:secondterm}
\vert \sqrt{\mu }e_\mu\vert \ge c_\delta \,  \mu^{-\min(s^*+\delta,2)/2+1/2} \quad \text{and} \quad   \sup_{0<T\le C \mu}\sqrt{\mu} \big\Vert \mathfrak{F}_\mu V M_{T,\mu}^{(\kappa)} V \mathfrak{F}_\mu^\dagger \big\Vert_{\mathrm{op}} \le C_\delta \mu^{-2\beta^*+1 + \delta}
\end{equation}
for any $\delta >0$ by Lemma~\ref{lem:4} resp.~\Cref{eq:case11} and \Cref{eq:case12} (recall $T_c = o(\mu)$ from above).
The error from the series expansion is of order $O(\mu^{-3\beta^* + 3/2+\delta})$ and the error from the perturbation argument is of order $O(\mu^{-4\beta^*+\min(s^*,2)/2 + 3/2+\delta})$
and hence dominates, since $\beta^* \leq \min(s^*,2)/2$.
 
 \Cref{eq:impleq1} is an implicit equation for $T_c$. By the second estimate in \Cref{eq:secondterm} 
 and since $T_c \to 0 $ as $\mu \to \infty$, we need to evaluate the limit of $\langle u\vert  \mathfrak{F}_\mu V M_{T,\mu}^{(\kappa)} V \mathfrak{F}_\mu^\dagger \vert  u\rangle$ as $T \to 0$. By the same arguments as used in Equation~(35) in \cite{hs081} (dominated convergence and Lipschitz continuity of the angular integrals), this yields
 \begin{equation*} 
\lim\limits_{T \to 0}\big\langle u \big\vert \mathfrak{F}_\mu V M_{T,\mu}^{(\kappa)} V \mathfrak{F}_\mu^\dagger \big\vert u \big\rangle = \big\langle u \big\vert\mathcal{W}_\mu^{(\kappa)}  \big\vert u \big\rangle  \,,
 \end{equation*}
 uniformly in $u \in L^2(\Sph^2)$ with $\Vert u \Vert_{L^2(\Sph^2)} = 1$, where $\mathcal{W}_\mu^{(\kappa)}$ was defined in \eqref{eq:defWmu}. Using that $T_c$ is exponentially small (in some positive power of $\mu$) as $\mu \to \infty$ by application of Theorem~\ref{thm:1} in combination with \Cref{elbound}, we obtain
 \begin{equation} \label{eq:Tconv}
\big\vert \big\langle u \big\vert \mathfrak{F}_\mu V M_{T_c,\mu}^{(\kappa)} V \mathfrak{F}_\mu^\dagger \big\vert u \big\rangle - \big\langle u \big\vert\mathcal{W}_\mu^{(\kappa)}  \big\vert u \big\rangle\big\vert \le C_D \mu^{-D}
 \end{equation}
 for any $D>0$, uniformly in $u \in L^2(\Sph^2)$ with $\Vert u \Vert_{L^2(\Sph^2)} = 1$. Combining the second estimate in \Cref{eq:secondterm} with \Cref{eq:Tconv} we thus get
 \begin{equation} \label{eq:Wmubound}
\left\Vert \mathcal{W}_\mu^{(\kappa)} \right\Vert_{\mathrm{op}} \le C_\delta \mu^{-2\beta^*+1/2 + \delta}
 \end{equation}
 for any $\delta > 0$. Since $\vert \sqrt{\mu} e_\mu \vert \ge c_\delta \,  \mu^{-\min(s^*+\delta,2)+1/2}$, we have that whenever $e_\mu <0$ also $b_\mu^{(\kappa)}<0$ for large enough $\mu$ (recall \Cref{eq:Bmu} and \Cref{eq:bmu}). In particular, combining Equations \eqref{eq:impleq1}, \eqref{eq:Tconv} and \eqref{eq:Wmubound}, we have shown that
 \begin{equation*}
\frac{1}{\sqrt{\mu}} m_\mu^{(\kappa)}(T_c) + \frac{\pi}{2 \sqrt{\mu} b_\mu^{(\kappa)}} = O(\mu^{-4\beta^* + 3 \min(s^*,2)/2+1/2+\delta})\,,
 \end{equation*}
for any $\delta > 0$. Since $4\beta^* - 3 \min(s^*,2)/2-1/2 > 0$ (see \Cref{eq:perttheoryestimate}), we conclude  \Cref{eq:thm2kappa} by employing the asymptotics \eqref{masymp}. 

In case that there exists $\mu_0 >0$ and $\mathcal{L} \subset \mathbb{N}_0$ with $\vert \mathcal{L} \vert < \infty$, such that for all $\mu \ge \mu_0$, the ground state space of $\mathcal{V}_\mu$ is contained in the finite--dimensional subspace 
\begin{equation*}
\mathcal{I}_\mathcal{L} := \mathrm{span} \left\{ Y_\ell^m : \ell \in \mathcal{L}, \, |m| \le \ell  \right\} 
\end{equation*} 
of $L^2(\Sph^2)$, spanned by the spherical harmonics with angular momentum $\ell \in \mathcal{L}$, we can drop condition (d) from Definition \ref{def:admpot}. In order to see this, take $Y_\ell^m$ with $\ell \in \mathcal{L}$ and $\vert m \vert \le \ell$ and estimate 
\begin{align*}
\left\Vert|V|^{1/2} \mathfrak{F}_\mu^\dagger Y_\ell^m\right\Vert_{L^2}^2  = C \int_{\R^3} |V(x)|  &\left\vert \int_{\Sph^2} e^{i\sqrt{\mu}p\cdot x} Y_\ell^m(p)\D \omega(p) \right\vert^2 \D x \\[1mm]
&= C \int_{\R^3} |V(x)|(j_\ell(\sqrt{\mu} |x|))^2 \D x \, \le \, C_{\ell, \delta} \, \mu^{-\min(s^*,2)+ \delta}\,,
\end{align*}
for any $\delta >0$. The second equality follows by the radiality of $V$ and the final estimate by the decay of spherical Bessel functions (see the first bound in Proposition~\ref{prop:2}~(iii)). Using \Cref{eq:case12} with $\beta^*_{3/2}$ instead of $\beta^*$ by means of Lemma \ref{lem:3}, this implies that 
 \begin{equation} \label{eq:specialcase}
\sup_{u \in \mathcal{I}_\mathcal{L}\,, \, \Vert u \Vert_{L^2} = 1}\big\vert \sqrt{\mu}\, \big\langle u \big\vert \mathfrak{F}_\mu V M_{T,\mu}^{(\kappa)} V \mathfrak{F}_\mu^\dagger \big\vert u \big\rangle\big\vert \le C_{\mathcal{L}, \delta} \, \mu^{-\beta^*_{3/2}-\min(s^*,2)/2+1 + \delta}
\end{equation}
for any $\delta >0$ (and $\mu$ large enough). Therefore, since $\beta^*_{3/2} \le \min(s^*,2)/2$, the error from the geometric expansion dominates the error from the perturbation theory in \Cref{eq:impleq1} and we get
 \begin{equation} \label{eq:impleq2}
\frac{1}{\sqrt{\mu}} m_\mu^{(\kappa)}(T_c) = \frac{-1}{\sqrt{\mu} e_\mu - \sqrt{\mu} \big\langle u \big\vert \mathfrak{F}_\mu V M_{T_c,\mu}^{(\kappa)} V \mathfrak{F}_\mu^\dagger \big\vert u \big\rangle + O(\mu^{-3\beta^*_{3/2} + 3/2+\delta})}\,,
\end{equation}
for any $\delta > 0$, instead. Moreover, using \Cref{eq:specialcase} and \Cref{eq:Tconv}, we get
\begin{equation} \label{eq:Wmuboundspecial}
\left\Vert \mathcal{W}_\mu^{(\kappa)}\big\vert_{\mathcal{I}_\mathcal{L}} \right\Vert_{\mathrm{op}} \le C_{\mathcal{L}, \delta} \mu^{-\beta^*_{3/2}-\min(s^*,2)/2+1/2 + \delta}
\end{equation}
for any $\delta >0$. By combining Equations \eqref{eq:impleq2}, \eqref{eq:Tconv} and \eqref{eq:Wmuboundspecial} with $\vert \sqrt{\mu} e_\mu \vert \ge c_\delta \,  \mu^{-\min(s^*+\delta,2)+1/2}$, we find
 \begin{equation*}
\frac{1}{\sqrt{\mu}} m_\mu^{(\kappa)}(T_c) + \frac{\pi}{2 \sqrt{\mu} b_\mu^{(\kappa)}} = O(\mu^{-3\beta^*_{3/2} +  \min(s^*,2)+1/2+\delta})
\end{equation*}
for any $\delta >0$, and the proof comes to an end in the same way as above by realizing that $3\beta^*_{3/2}-\min(s^*,2)-1/2>0$. 
\end{proof}
\begin{proof}[Proof of Lemma \ref{lem:6}]
The proof follows a similar perturbation theoretic argument as in the proof of Lemma \ref{lem:5}. We first assume condition (d) from Definition \ref{def:admpot} and discuss the changes for the special case afterwards.  To begin with, we derive a bound on $\Vert \mathcal{W}_\mu\Vert_{\mathrm{op}}$. For this purpose, we take a normalized $u \in L^2(\Sph^2)$ and estimate
\begin{align*}
\big\vert \big\langle u \big\vert\mathcal{W}_\mu  \big\vert u \big\rangle \big\vert \ \le \ &\big\vert \big\langle u \big\vert \mathcal{W}_\mu \big\vert u \big\rangle - \big\langle u \big \vert\mathcal{W}_\mu^{(\kappa)}  \big\vert u \big\rangle \big\vert \\ &+ \big\vert \big\langle u \big\vert \mathfrak{F}_\mu V M_{T_c,\mu}^{(\kappa)} V \mathfrak{F}_\mu^\dagger  \big\vert u \big\rangle - \big\langle u \big \vert \mathcal{W}_\mu^{(\kappa)}  \big\vert u \big\rangle\big\vert + \big\vert \big\langle u \big\vert \mathfrak{F}_\mu V M_{T_c,\mu}^{(\kappa)} V \mathfrak{F}_\mu^\dagger \big\vert u \big\rangle\big\vert\,.
\end{align*}
The second term is smaller than any inverse power of $\mu$ by \Cref{eq:Tconv}. Using \Cref{eq:case12} and \Cref{eq:case11}, the third term is bounded by $\mu^{-2\beta^* + 1/2 + \delta}$ for any $\delta >0$, uniformly in $u \in L^2(\Sph^2)$. Since
\begin{equation} \label{eq:differencekappa0}
\big\langle u \big\vert \mathcal{W}_\mu \big\vert u \big\rangle - \big\langle u \big \vert\mathcal{W}_\mu^{(\kappa)}  \big\vert u \big\rangle  
= \sqrt{\mu} \int_{0}^{\infty} \hspace{-2mm}\D \vert p \vert \left(1- \frac{\vert p\vert^2 }{\vert p \vert^2 + \kappa^2}\right) \left\Vert \mathcal{V}_\mu u \right\Vert_{L^2}^2 = \kappa \, \frac{\pi}{2} \sqrt{\mu} \left\Vert \mathcal{V}_\mu u \right\Vert_{L^2}^2\,,
\end{equation}
we infer by means of \Cref{eq:case12} that also the first term is bounded by $\mu^{-2\beta^* + 1/2 + \delta}$, uniformly in $u \in L^2(\Sph^2)$, and we thus have 
\begin{equation} \label{eq:Wmubound2}
\left\Vert \mathcal{W}_\mu\right\Vert_{\mathrm{op}} \le C_\delta \mu^{-2\beta^*+1/2 + \delta}
\end{equation}
for any $\delta > 0$. In particular, since $\vert \sqrt{\mu} e_\mu \vert \ge c_\delta \,  \mu^{-\min(s^*+\delta,2)+1/2}$ for any $\delta > 0$, this shows that, whenever $e_\mu <0$ also $b_\mu<0$ for large enough $\mu$. Moreover, using $\vert \sqrt{\mu} e_\mu \vert \ge c_\delta \,  \mu^{-\min(s^*+\delta,2)+1/2}$ together with \Cref{eq:Wmubound} and \Cref{eq:Wmubound2}, first order perturbation theory implies
\begin{align}
&\frac{\pi}{2\sqrt{\mu} b_\mu} - \frac{\pi}{2\sqrt{\mu} b_\mu^{(\kappa)}}  = \frac{\pi}{2}\frac{b_\mu^{(\kappa)} - b_\mu}{\sqrt{\mu}\,  b_\mu^{(\kappa)} b_\mu } \nonumber \\
= &\, \frac{ \big(e_\mu -  \big\langle u \big\vert \mathcal{W}_\mu^{(\kappa)} \big\vert u \big\rangle\big) - \big(e_\mu - \big\langle u' \big \vert\mathcal{W}_\mu \big\vert u' \big\rangle\big)  + O(\mu^{-4\beta^*+\min(s^*,2)/2 + 3/2+\delta})}{\sqrt{\mu} e_\mu^2 + O(\mu^{-2\beta^* + 1+\delta})} \label{eq:comparison}\\ \nonumber
= &\, \kappa \, \frac{\pi}{2} + O(\mu^{-4\beta^* + 3 \min(s^*,2)/2+1/2+\delta})\,. 
\end{align}
As in the proof of Lemma \ref{lem:5}, $u$ resp.~$u'$ is a (the) normalized eigenfunction corresponding to the lowest eigenvalue $e_\mu$ of $\mathcal{V}_\mu$. In case of (finite!) degeneracy, one has to choose the ground state $u$ resp.~$u'$ of $\mathcal{V}_\mu$ that minimizes the second term in each bracket $(\cdots)$ in \Cref{eq:comparison}. \textit{A priori}, $u$ and $u'$ could be different. But, by application of \Cref{eq:differencekappa0} we get that $\mathcal{W}_\mu$ and $\mathcal{W}_\mu^{(\kappa)}$ differ only by the constant $(\kappa \pi \sqrt{\mu}e_\mu^2)/2$ on the ground state space of $\mathcal{V}_\mu$. Therefore, $u = u'$ and the last equality in \Cref{eq:comparison} follows by \Cref{eq:differencekappa0} in combination with $\vert \sqrt{\mu} e_\mu \vert \ge c_\delta \,  \mu^{-\min(s^*+\delta,2)+1/2}$. Since $4\beta^* - 3 \min(s^*,2)/2-1/2 > 0$ (see \Cref{eq:perttheoryestimate}), \Cref{eq:comparison} implies \Cref{eq:kappa0}.

In case that there exists $\mu_0 >0$ and $\mathcal{L} \subset \mathbb{N}_0$ with $\vert \mathcal{L} \vert < \infty$, such that for all $\mu \ge \mu_0$, the ground state space of $\mathcal{V}_\mu$ is contained in the finite--dimensional subspace 
\begin{equation*}
\mathcal{I}_\mathcal{L} := \mathrm{span} \left\{ Y_\ell^m : \ell \in \mathcal{L}, \, |m| \le \ell  \right\} 
\end{equation*} 
of $L^2(\Sph^2)$, spanned by the spherical harmonics with angular momentum $\ell \in \mathcal{L}$, we can drop condition (d) from Definition \ref{def:admpot}. In order to see this, we use \Cref{eq:eigenvalues} and estimate 
\begin{equation*}
\left\Vert \mathcal{V}_\mu\big\vert_{\mathcal{I}_\mathcal{L}}  \right\Vert_{\mathrm{op}} = \sup_{\ell \in \mathcal{L}} \left\vert \frac{1}{2\pi^2} \int_{\R^3} V(x) \left(j_\ell(\sqrt{\mu}\vert x \vert)\right)^2 \D x \right\vert \le C_{\mathcal{L}, \delta} \, \mu^{-\min(s^*,2)+ \delta}
\end{equation*}
for any $\delta >0$ (and $\mu$ large enough) by means of Proposition~\ref{prop:2}~(iii). Combining this with \Cref{eq:specialcase} and using $\beta^*_{3/2} \le \min(s^*,2)/2$, we get by the same argument as above that
\begin{equation} \label{eq:Wmubound2special}
\left\Vert \mathcal{W}_\mu\big\vert_{\mathcal{I}_\mathcal{L}} \right\Vert_{\mathrm{op}} \le C_{\mathcal{L}, \delta} \mu^{-\beta^*_{3/2}-\min(s^*,2)/2+1/2 + \delta}
\end{equation}
for any $\delta >0$. Using first order perturbation theory, \Cref{eq:Wmubound2special} and \Cref{eq:differencekappa0} together with $\vert \sqrt{\mu} e_\mu \vert \ge c_\delta \,  \mu^{-\min(s^*+\delta,2)+1/2}$ imply
\begin{equation*}
\frac{\pi}{2\sqrt{\mu} b_\mu} - \frac{\pi}{2\sqrt{\mu} b_\mu^{(\kappa)}} = \kappa \, \frac{\pi}{2} + O(\mu^{-3\beta^*_{3/2} + \min(s^*,2)+1/2+\delta})
\end{equation*}
for any $\delta >0$. Since $3\beta^*_{3/2} - \min(s^*,2)-1/2>0$ we conclude the desired. 
\end{proof}
\noindent Finally, we give the proofs of Lemma \ref{lem:3} and Lemma \ref{lem:4}. 
\begin{proof} [Proof of Lemma \ref{lem:3}]
	For $s^* \in (1,5/3]$ the statement easily follows from the uniform decay of spherical Bessel functions (see Proposition~\ref{prop:2}~(iii)). For $s^* > 5/3$ choose
	\begin{equation}
	\alpha = \max\left( \frac{5q-7}{3qs^* - 7}, \, \frac{5q-7}{6q-7}\right)\in (0,1) \label{eq:alphabounds}
	\end{equation}
	and for (small) $\delta > 0$ set $s:= \min(s^*,2) - \delta/\alpha$. Recall that $q=p/(p-1)$ denotes the dual of $p \in [3/2,9/4]$. We now employ Hölder's inequality to obtain
	\begin{align*}
	&\sup_{\ell \in \mathbb{N}_0}\int_{\R^3} \D x \vert V (x)\vert \vert j_\ell(\sqrt{\mu}\vert x \vert )\vert^{2-\varepsilon} \\ \le & \, C \left\Vert \frac{V}{\vert \cdot \vert^s}\right\Vert_{L^1}^\alpha \Vert V \Vert_{L^{p}}^{1-\alpha} \sup_{\ell \in \mathbb{N}_0} \left(\int_{0}^{\infty} \D x \, x^{\frac{q\alpha s}{1-\alpha} + 2} \, \vert j_\ell (\sqrt{\mu}x) \vert^{\frac{q}{1-\alpha}(2-\epsi)}\right)^{\frac{1-\alpha}{q}}\\
	\le \, & C \mu^{- \frac{\alpha s +3 (1-\alpha)/q }{2}}\left\Vert \frac{V}{\vert \cdot \vert^s}\right\Vert_{L^1}^\alpha \Vert V \Vert_{L^{p}}^{1-\alpha} \sup_{\ell \in \mathbb{N}_0} \left(\int_{0}^{\infty} \D x \, x^{\frac{q\alpha s}{1-\alpha} + 2} \, \vert j_\ell (x) \vert^{\frac{q}{1-\alpha}(2-\epsi)}\right)^{\frac{1-\alpha}{q}}\,.
	\end{align*}
	For $\epsi(\delta)>0$ small enough, the integral is finite by the uniform $L^p$--integrability of spherical Bessel functions (see Proposition~\ref{prop:2}~(iv)) since $\alpha <(5q-7)/(3qs-7)$ and thus the claim follows since $\frac{\alpha s +3 (1-\alpha)/q }{2} = \beta^*_p - \delta/2 $ (cf.~\Cref{eq:betabound} for the definition of $\beta_p^*$, and \Cref{eq:alphabounds}). 
\end{proof}
\begin{proof}[Proof of Lemma \ref{lem:4}]
	To begin with the proof, we have two important observations. 
	
	First, recall the definition of $s^*_{\pm}$ from \Cref{eq:defs}. We aim to prove that $r^*_{\pm} = s^*_{\pm}$, where
	\begin{equation*} 
	r^*_{\pm} := \sup \left\{ r \ge 0 : \lim\limits_{\varepsilon \to 0} \frac{1}{\varepsilon^r} \int_{B_\epsi} V_{\pm}(x) \D x = 0 \right\}\,.
	\end{equation*}
For this purpose, we define
	\begin{equation*}
s_{\pm}^*(a) := \sup \left\{ s \ge 0 : \vert \cdot \vert^{-s} V_{\pm}\vert_{B_a}^* \in L^1(\R^3)\right\} 
\end{equation*}
and
	\begin{equation*} 
	r^*_{\pm}(a) :=  \sup \left\{ r \ge 0 : \lim\limits_{\varepsilon \to 0} \frac{1}{\varepsilon^r} \int_{B_\epsi}  V_{\pm}\vert_{B_a}^*(x) \, \D x = 0 \right\} 
	\end{equation*}
	for the same $a>0$, for which we assumed that $r^*_\pm = r^*_\pm(a)$ in Definition \ref{def:admpot}. 

 Note that $r^*_\pm \ge s^*_\pm$ by definition. Using that $\vert \cdot \vert^{-s}$ is equal to its symmetric decreasing rearrangement, we can employ the basic rearrangement inequality \cite[Thm.~3.4]{liebloss} to obtain $s^*_\pm \ge s^*_\pm(a)$. Therefore, since $r^*_\pm = r^*_\pm(a)$ by assumption, we have
	\begin{equation*}
	r^*_\pm(a) = r^*_\pm \ge s^*_\pm \ge s^*_\pm(a) \,.
	\end{equation*}
	In order to see $r^*_\pm = s^*_\pm$ it is sufficient to prove that $s^*_\pm(a) \ge r_\pm^*(a) $. Assume the contrary, i.e.~$s_\pm^*(a) < r_\pm^*(a) $, and
let $r, \, r+ \delta \in (s_\pm^*(a),r_\pm^*(a))$ for some $\delta >0$. We denote $V_{\pm,a}^* \equiv V_\pm\vert_{B_a}^* $ for short. By definition of $s_\pm^*$ and $r_\pm^*$, we thus have
	\begin{equation} \label{eq:rstrans}
	\int_{B_\epsi} \frac{V_{\pm,a}^*(x)}{\vert x \vert^r} \D x \ge c \qquad \text{and} \qquad \int_{B_\epsi}  V_{\pm,a}^*(x) \D x = o(\epsi^{r+\delta})\,.
	\end{equation}
The first integral actually equals infinity, but we only need that it is uniformly bounded from below by some $c>0$. 	Since $\vert V_{\pm,a} \vert^*$ is symmetric--decreasing and thus one--sided limits exist, the auxiliary quantity
\begin{equation*}
t_\pm^*(a) := \inf \left\{ t \ge 0 : \lim\limits_{\vert x \vert \to 0} \vert x \vert^t \,   V_{\pm,a}^*(x) = 0  \right\}
\end{equation*}
is well defined. By definition of $t_\pm^*(a)$ we thus get
	\begin{equation*}
\frac{c_\nu}{\vert x \vert^{t_\pm^*(a)-\nu}} \le V_{\pm,a}^*(x) \le \frac{C_\nu}{\vert x \vert^{t_\pm^*(a)+\nu}} 
	\end{equation*}
	for any $\nu >0$ and $\vert x \vert$ small enough. Inserting this in \Cref{eq:rstrans} we arrive at
	\begin{equation*}
	  \epsi^{3-t_\pm^*(a) -r-\nu}   \ge c_\nu \qquad \text{and} \qquad \epsi^{3-t_\pm^*(a) -r-\delta + \nu} \le C_\nu
	\end{equation*}
	which yields a contradiction by choosing $\nu \in (0,\delta/2)$. Therefore, $r_\pm^*(a) = s_\pm^*(a)$, which proves that $r_\pm^* = s_\pm^*$. 
	
	Second, note that for any $f \in L^1(\R^3)$ we have 
	\begin{equation*}
	\int_{\R^3} f(x) (\sin(n\vert x \vert ))^2 \D x = \frac{1}{2}\int_{\R^3} f(x) (1-\cos(2n\vert x \vert )) \D x  \longrightarrow \frac{1}{2} \int_{\R^3} f(x) \D x 
	\end{equation*}
	as $n \to \infty$ by the Riemann--Lebesgue Lemma. 
	
	In order to prove Lemma \ref{lem:4}, we study the asymptotic behavior of the integral 
	\begin{equation*}
v_\mu:= \int_{\R^3} V(x) \left(\frac{\sin(\sqrt{\mu}\vert x \vert)}{\sqrt{\mu}\vert x \vert}\right)^2 \D x 
	\end{equation*}
	in three different cases. 
	
	\textit{Case 1.} If $\vert \cdot \vert^{-2}V\in L^1(\R^3)$, we get by our second observation that 
	\begin{equation*}
v_\mu = \mu^{-1} \left( \frac{1}{2}\int_{\R^3} \frac{V(x)}{\vert x \vert^2} \, \D x + o(1)\right) \le - c \mu^{-1}\,,
	\end{equation*}
	which immediately proves the claim. 
	
		\textit{Case 2.} If $\vert \cdot \vert^{-2}V \notin L^1(\R^3)$ and $s^* < 2$ we take some $r \in (0,1/2)$ and estimate
	\begin{align*}
	\int_{\R^3} V(x) &\left(\frac{\sin(\sqrt{\mu}\vert x \vert)}{\sqrt{\mu}\vert x \vert}\right)^2 \D x  \\
	&= \int_{B_{r}} V_+(x) \left(\frac{\sin(\sqrt{\mu}\vert x \vert)}{\sqrt{\mu}\vert x \vert}\right)^2 \D x - \int_{B_{r}} V_-(x) \left(\frac{\sin(\sqrt{\mu}\vert x \vert)}{\sqrt{\mu}\vert x \vert}\right)^2 \D x + O(\mu^{-1})\,.
	\end{align*}
	The first term can be bounded by $\mu^{-s^*_+/2+\delta}$ for any $\delta > 0$. The second term can be estimated from below as
	\begin{align*}
\int_{B_r} V_-(x)  \left(\frac{\sin(\sqrt{\mu}\vert x \vert)}{\sqrt{\mu}\vert x \vert}\right)^2 \D x  \ge 	\int_{B_{\frac{r}{\sqrt{\mu}}}} V_-(x)  &\left(\frac{\sin(\sqrt{\mu}\vert x \vert)}{\sqrt{\mu}\vert x \vert}\right)^2 \D x \\
&\ge c_r\int_{B_{\frac{r}{\sqrt{\mu}}}}  V_- (x) \D x \ge c_{r,\delta} \mu^{-s^*_-/2-\delta}
	\end{align*}
	for any $\delta > 0$. Since $s^* = s^*_- < s^*_+$, we get that $v_\mu \le - c_\delta \mu^{-\min(s^*+\delta, 2)/2}$ for any $\delta > 0$.
	
	 	\textit{Case 3.}	If $\vert \cdot \vert^{-2}V \notin L^1(\R^3)$ and $s^* = 2$ we have that $\vert \cdot \vert^{-2}V_+ \in L^1(\R^3)$ but $\vert \cdot \vert^{-2}V_- \notin L^1(\R^3)$ since ${s^* = s^*_- < s^*_+}$. On the one hand, this implies that 
	 	\begin{equation*}
	\int_{\R^3} V_+(x) \left(\frac{\sin(\sqrt{\mu}\vert x \vert)}{\sqrt{\mu}\vert x \vert}\right)^2 \D x \le K \mu^{-1}
	 	\end{equation*}
	 	for some $K>0$ by means of our second observation. On the other hand, let $r >0$ and estimate 
	 	\begin{align*}
	\mu \int_{\R^3} V_-(x) \left(\frac{\sin(\sqrt{\mu}\vert x \vert)}{\sqrt{\mu}\vert x \vert}\right)^2 \D x \ge \mu \int_{B_{r}^c} \hspace{-1mm}V_-(x) \left(\frac{\sin(\sqrt{\mu}\vert x \vert)}{\sqrt{\mu}\vert x \vert}\right)^2 \D x  
	\stackrel{\mu \to \infty}{\longrightarrow} \frac{1}{2}\int_{B_{r}^c} \frac{V_-(x)}{\vert x \vert^2 }  \D x\,.
	 	\end{align*}
	 	By taking $r \to 0$ the right hand side can be made arbitrarily large, in particular greater than $K$. This implies that $v_\mu \le - C \mu^{-1}$ for any $C>0$. 
\end{proof}

 \subsection{Properties of spherical Bessel functions}
 \begin{prop}{\rm (Properties of spherical Bessel functions \cite{abramowitz, landau,Besselint})} \label{prop:2} \\
The spherical Bessel functions $(j_\ell)_{\ell \in \mathbb{N}_0}$ satisfy the following properties: 
\begin{itemize}
	\item[(i)] uniform boundedness, i.e.~$\sup_{\ell \in \mathbb{N}_0} \sup_{x\ge 0}\vert j_\ell(x)\vert  \le 1$,
		\item[(ii)] uniform Lipschitz continuity, i.e.~$\sup_{\ell \in \mathbb{N}_0} \sup_{x\ge 0}\vert j'_\ell(x)\vert  \le 1$,
	\item[(iii)] (uniform) decay, i.e.~for every $\ell \in \mathbb{N}_0$, we have $\sup_{x\ge 0}\vert x\,  j_\ell(x)\vert\le C_\ell$ for some $C_\ell >0$, and $\sup_{\ell \in \mathbb{N}_0} \sup_{x\ge 0}\vert x^{5/6} j_\ell(x)\vert\le C$ for some universal $C>0$,
	\item[(iv)] uniform $L^p$--integrability, i.e.~for $p \in (0,\infty)$ and $a \in (-1,p-1)$ if $p \in (0,4]$ or $a \in (-1,5p/6-1/3)$ if $p \in (4,\infty)$, we have
	\begin{equation*}
	\sup_{\ell \in \mathbb{N}_0} \int_{0}^\infty \vert j_\ell(x) \vert^p x^a \D x \le C_{p,a}
	\end{equation*}
	for some universal constant $C_{p,a}>0$. 
\end{itemize}
 \end{prop}
\begin{proof}
	The first statement (i) is an elementary property of the spherical Bessel functions. The second statement (ii) follows from the uniform boundedness in (i) and the recursion relation \cite[Eq.~10.1.20]{abramowitz}
	\begin{equation*}
	j'_\ell = \frac{1}{2\ell + 1}\left(\ell j_{\ell -1} - (\ell + 1)j_{\ell+1}\right) \,.
	\end{equation*}
	By noticing that $j_\ell(x) = \sqrt{\pi/(2x)} J_{\ell + 1/2}(x)$, the third (iii)  and the fourth statement (iv) are easy consequence of \cite[Eq.~9.2.1]{abramowitz}, \cite[Eq.~1]{landau}, and \cite[Eq.~3]{Besselint}, respectively, where analogous estimates for the standard Bessel functions $J_\nu$ with $\nu \ge 0$ are proven. 
\end{proof}

\subsection{Proof of Equation \eqref{eq:gontierbound}}
\begin{proof}[Proof of \Cref{eq:gontierbound}]
We note that $K_{T,\mu}(p) + V(x) \ge 0$ is equivalent to $K_{T/\mu,1}(p) + \frac{1}{\mu} V(x/\sqrt{\mu}) \ge 0$ and estimate
\begin{align*}
K_{T/\mu,1}(p) \, +\,  \frac{1}{\mu} V(x/\sqrt{\mu}) 
&\ge \frac{1}{2}\left( \vert p^2 - 1 \vert + \frac{2T}{\mu}\right) - \frac{1}{\mu} V_-(x/\sqrt{\mu}) \\
&\ge \frac{1}{2}\left( \vert p^2 - 1 \vert + \frac{2T}{\mu} - \frac{2}{\mu} V_-(x/\sqrt{\mu}) \left(\E^{-m\vert x \vert} + m\vert x \vert \right)\right) \\
&\ge \frac{1}{2}\left( \vert p^2 - 1 \vert + \frac{2T}{\mu} - \frac{2}{\sqrt{\mu}} \Vert  \vert \cdot \vert V \Vert_{L^\infty} \left(\frac{\E^{-m\vert x \vert}}{\vert x \vert } + m \right)\right) 
\end{align*}
for any $m>0$. By definition of $T_c$, we have the bound
\begin{equation*}
T_c \le - \frac{\mu}{2} \, \mathrm{inf\,spec}\left(\vert p^2 - 1 \vert  - \frac{2}{\sqrt{\mu}} \Vert  \vert \cdot \vert V \Vert_{L^\infty} \left(\frac{\E^{-m\vert x \vert}}{\vert x \vert } + m \right)\right)\,.
\end{equation*} 
After taking $m = (\mathrm{const.}) \mu^{1/4} \E^{- \sqrt{\pi/(2 \Vert  \vert \cdot \vert V \Vert_{L^\infty})} \, \mu^{1/4}}$ and using the estimate above Equation (15) in \cite{gontier}, we get
\begin{equation*}
T_c \lesssim \mu \exp\left(- \sqrt{\frac{\pi}{2 \Vert  \vert \cdot \vert V\Vert_{L^\infty}}} \mu^{1/4}\right) \,. \qedhere
\end{equation*}
\end{proof}

\section*{Declarations}
\textbf{Acknowledgments.} 
I am very grateful to Robert Seiringer for his guidance during this project and for many valuable comments on an earlier version of the manuscript. Moreover, I would like to thank Asbj{\o}rn B\ae kgaard Lauritsen for many helpful discussions and comments, pointing out the reference \cite{Besselint} and for his involvement in a closely related joint project \cite{HLenergygap}. Finally, I am grateful to Christian Hainzl for valuable comments on an earlier version of the manuscript and Andreas Deuchert for interesting discussions. 
\\[3mm]
\textbf{Funding.} Partial financial support by the ERC Advanced Grant ``RMTBeyond” No.~101020331 is gratefully acknowledged. \\[3mm]
\textbf{Conflicts of Interest.} The author has no relevant financial or non-financial interests to disclose.
\\[3mm]
\textbf{Availability of data and materials.} Not applicable.


\begin{thebibliography}{00}
   	\bibitem{abramowitz} M.~Abramowitz, I.~A.~Stegun (eds.). \textit{Handbook of mathematical functions with formulas, graphs, and mathematical tables.}  US Government printing office, 10th edition, 1972.
   	
   	\bibitem{bcs} J.~Bardeen, L.~N.~Cooper, J.~Schrieffer. Theory of Superconductivity.  \textit{Phys. Rev.} 108, 1175--1204, 1957. 
   	
   	\bibitem{graphene} Y.~Cao \textit{et al.} Unconventional superconductivity in Magic--Angle Graphene Superlattices. \textit{Nature} 556, 43–-50, 2018.
   	
   	\bibitem{cueninmerz} J.--C.~Cuenin, K.~Merz. Weak Coupling Limit for Schrödinger Operators with Degenerate Kinetic Energy for a Large Class of Potentials. \textit{Lett. Math. Phys.} 111, 46, 2021.
   	
   	\bibitem{cuprates} E.~Dagotto. Correlated electrons in High--Temperature Superconductors. \textit{Reviews of Modern Physics} 66, 763--841, 1994.
   	
   	\bibitem{fhns} R.~L.~Frank, C.~Hainzl, S.~Naboko, R.~Seiringer. The Critical Temperature for the BCS Equation at Weak Coupling. \textit{The Journal of Geometric Analysis} 17, 559--567, 2007. 
   	
\bibitem{gontier} D.~Gontier, C.~Hainzl, M.~Lewin. Lower Bound on the Hartree--Fock Energy of the Electron Gas. \textit{Phys. Rev. A} 99, 052501, 2019.

\bibitem{gorkov}  L.~P.~Gor’kov, T.~K.~Melik--Barkhudarov. Contributions to the Theory of Superfluidity in an Imperfect Fermi Gas. \textit{Soviet Physics JETP} 13, 1018, 1961. 

\bibitem{hhss} C.~Hainzl, E.~Hamza, R.~Seiringer, J.~P.~Solovej. The BCS Functional for General Pair Interaction. \textit{Comm. Math. Phys.} 281, 349--367, 2008. 

\bibitem{hs081} C.~Hainzl, R.~Seiringer. Critical Temperature and Energy Gap for the BCS Equation. \textit{Phys. Rev. B}, 77, 184517, 2008.

\bibitem{hs08} C.~Hainzl, R.~Seiringer. The BCS Critical Temperature for Potentials with Negative Scattering Length. \textit{Lett. Math. Phys.} 84, 99--107, 2008. 

\bibitem{hs10} C.~Hainzl, R.~Seiringer. Asymptotic Behavior of Eigenvalues of Schrödinger Type Operators with Degenerate Kinetic Energy. \textit{Math. Nachr.} 283, 489--499, 2010. 

\bibitem{hs15} C.~Hainzl, R.~Seiringer. The Bardeen-Cooper--Schrieffer Functional of Superconductivity and its Mathematical Properties. \textit{J. Math. Phys.} 57, 021101, 2016. 

\bibitem{HLenergygap} J.~Henheik, A.~B.~Lauritsen. The BCS Energy Gap at High Density. \textit{arXiv:} \href{https://arxiv.org/pdf/2106.02028.pdf}{2106.02028}, 2021. 

\bibitem{srtio} C.~S.~Koonce \textit{et al.} Superconducting Transition Temperatures of Semiconducting SrTiO$_3$. \textit{Phys. Rev.} 163, 380--390, 1967. 

\bibitem{landau} L.~J.~Landau. Bessel Functions: Monotonicity and Bounds.  \textit{Journal of the London Math. Soc.}, 61, 197--215, 2000. 

\bibitem{langmann} E.~Langmann, C.~Triola, A.~V.~Balatasky. Ubiquity of Superconducting Domes in the Bardeen--Cooper--Schrieffer Theory with Finite--Range Potentials. \textit{Phys. Rev. Lett.} 122, 157001, 2019. 

\bibitem{laptev} A.~Laptev, O.~Safronov, T.~Weidl. Bound State Asymptotics for Elliptic Operators with Strongly Degenerated Symbols. In \textit{Nonlinear problems in mathematical physics and related topics I}, 233--245, 2002. 

\bibitem{lauritsen} A.~B.~Lauritsen. The BCS Energy Gap at Low Density. \textit{Lett. Math. Phys.} 111, 20, 2021. 

\bibitem{liebloss} E.~H.~Lieb, M.~Loss. Analysis. \textit{Amer. Math. Soc.}, 2nd edition, 2001. 

\bibitem{heavyfermion} N.~D.~Mathur \textit{et al.}. Magnetically Mediated Superconductivity in Heavy Fermion Compounds. \textit{Nature} 394, 39--43, 1998. 

\bibitem{pnictides} T.~Shibauchi \textit{et al.}. A Quantum Critical Point Lying Beneath the Superconducting Dome in Iron Pnictides. \textit{Ann. Rev. Cond. Matt. Phys.} 5, 113--135, 2014. 

\bibitem{Besselint} K.~Stempak. A Weighted Uniform $L^p$--estimate of Bessel Functions: A Note on a Paper of Guo. \textit{Proc. Amer. Math. Soc.} 128, 2943--2945, 2000. 

\bibitem{Fourier} E.~O.~Tuck. On Positivity of Fourier Transform. \textit{Bull. Austral. Math. Soc.} 74, 133--138, 2006.

\bibitem{bandins} J.~T.~Ye \textit{et al.} Superconducting Dome in a Gate--Tuned Band Insulator. \textit{Science} 338, 1193--1196, 2012. 
 \end{thebibliography}
\end{document}